\documentclass[a4paper,aps,pra,showpacs,twocolumn,superscriptaddress]{revtex4-2} 

\usepackage[utf8]{inputenc}  
\usepackage[T1]{fontenc}     
\usepackage[american]{babel}  
\usepackage[svgnames,dvipsnames]{xcolor}
\usepackage{diagbox}
\usepackage[percent]{overpic}
\usepackage[colorlinks=true,citecolor=DarkBlue,linkcolor=DarkBlue,urlcolor=DarkBlue,linktocpage,unicode]{hyperref}
\usepackage{amsmath,amssymb,amsthm,bm,mathtools,amsfonts,mathrsfs,bbm,dsfont}
\usepackage{centernot}
\usepackage{adjustbox}
\newtheorem{theorem}{Theorem}
\newtheorem{observation}[theorem]{Observation}

\newcommand{\tr}{{\mathrm{tr}}}

\newcommand{\eins}{\mathbbm{1}}
\newcommand{\swap}{\mathbb{S}}

\renewcommand{\vec}[1]{\ensuremath{\boldsymbol{#1}}}
\usepackage{braket}
\usepackage[normalem]{ulem}
\usepackage{comment}
\usepackage{svg}

\begin{document}
\title{Characterizing quantum channels from local-unitary invariants}

\author{Salwa Shaglel\hyperlink{email1}{\textsuperscript{*}}}
\affiliation{Institute for Quantum Inspired and Quantum Optimization, Hamburg University of Technology, Blohmstraße 15, 21079 Hamburg, Germany}

\author{Satoya Imai\hyperlink{email2}{\textsuperscript{\textdagger}}}
\affiliation{Institute of Systems and Information Engineering, University of Tsukuba, Tsukuba, Ibaraki 305-8573, Japan}
\affiliation{Center for Artificial Intelligence Research (C-AIR), University of Tsukuba, Tsukuba, Ibaraki 305-8577, Japan}

\date{\today}

\begin{abstract}
    We develop systematic frameworks for characterizing the entanglement properties of two-qubit channels beyond unitary settings. We introduce averaged local-unitary invariants, referred to as moments, obtained from Haar integrals over input states or unitaries. These moments provide computable descriptions of how a quantum channel can create, preserve, or destroy bipartite entanglement. We first show that second-order moments yield criteria for non-entangling and entanglement-breaking channels, which allow us to detect entanglement-creating and entanglement-preserving channels. We then demonstrate that higher-order moments can capture additional information and distinguish channels beyond second-order moments alone. Finally, we show that combinations of moments associated with different channel families improve the discrimination of locally inequivalent two-qubit unitaries.
\end{abstract}
\maketitle

\section{Introduction}
Precise control and manipulation of quantum systems are essential for quantum information processing tasks, such as communication \cite{teleportation1993}, computation \cite{Jozsa_2003}, and cryptography \cite{Gisin2002}. In practical implementations, quantum systems are inevitably coupled to their surrounding environments and are thus subject to noise and decoherence. In particular, as entanglement is a key resource in this field, understanding how quantum operations affect and transform entanglement is crucial for characterizing fundamental capabilities and limitations of information transmission through quantum channels.

A major research line in this direction has focused on the entanglement-creating ability of unitary channels. In this context, the notion of entangling power has been introduced to quantify the average amount of entanglement that a unitary can create from product states \cite{zanardi2000}. This framework has been extended to more general bipartite \cite{Wang_2002} and multipartite cases \cite{Scott_2004, Linowski_2020}. Related studies have yielded a detailed geometric classification of two-qubit unitary gates \cite{Zhang_2003, Reza2004, Watts_2013, Balakrishnan_2009, Balakrishnan_2010, Musz_2013}, as well as complementary notions such as gate typicality \cite{Jonnadula2017} and entangling-power deviation \cite{Cho_2026}.

Another research line has characterized the ability of noisy quantum channels that destroy entanglement. Unlike unitary channels, decoherence can degrade quantum correlations and eliminate them in the worst case, rendering the resulting dynamics classically simulable. Such entanglement-breaking channels have been studied through canonical parameterizations and separability properties of their associated Choi states \cite{Horodecki_2003, Ruskai_2003}. This line has been further developed in studies of infinite-dimensional systems \cite{Kan_2013,Muoi_2025}, entanglement-breaking channels under repeated actions \cite{Pasquale_2012,Lami_2016}, and the emergence of non-classical temporal correlations \cite{Vieira_2025}.

Despite these developments, the present results remain incomplete in several aspects. First, the entangling power is restricted to unitary channels and therefore cannot be applied to general non-unitary channels. This limitation prevents analyzing the entanglement-creating ability of nearly unitary channels subject to even small amounts of noise, which are ubiquitous in realistic quantum devices. Second, the entangling power is not injective: distinct locally inequivalent unitaries can possess the same value. As a result, it does not provide a full characterization of the entanglement-related properties of quantum operations. Third, existing approaches to entanglement-breaking channels are primarily qualitative, yielding only a binary (yes-no) classification of whether a channel is entanglement-breaking. Thus, they do not offer insight into the extent to which a channel approaches or departs from the entanglement-breaking regime.

In this manuscript, we address these issues by developing systematic frameworks for characterizing the entanglement properties of two-qubit channels. Our approach is based on averaged local-unitary invariants obtained by integrating over quantum states or unitaries with respect to the Haar measure, which we refer to as \emph{moments}. These quantities are applicable to both unitary and non-unitary cases and are evaluated without resorting to computationally demanding optimization procedures. Within these frameworks, we characterize the abilities of quantum channels to create, preserve, and break entanglement. In particular, we derive criteria for non-entangling and entanglement-breaking channels based on second-order moments. Furthermore, we show that higher-order moments can reveal additional features of quantum channels that cannot be accessed within second-order descriptions. Finally, we demonstrate that combining moments of different channel families can further enhance discrimination of two-qubit unitaries.

\section{Local-unitary invariants}
For a two-particle quantum state $\rho_{AB}$, a function $f(\rho_{AB})$ is called \emph{local-unitary (LU) invariant} if it satisfies $f(\rho_{AB}) = f (V_A \otimes V_B \rho_{AB} V_A^\dagger \otimes V_B^\dagger)$ for any local unitary $V_A \otimes V_B$. So far, LU invariants have been used to analyze quantum states \cite{Grassl_1998, Linden_1999, Makhlin2000, Sudbery_2001, Kraus_2010, Jing_2015}. In particular, a complete characterization of two-qubit states via $18$ LU invariants has been provided by Makhlin \cite{Makhlin2000}. Two examples of such LU invariants are
\begin{subequations}
    \begin{align}
    I_2 (\rho_{AB})
    &\vcentcolon = \sum_{i,j=1,2,3} T_{i,j}^2, \label{Eq:I2}
    \\
    I_4 (\rho_{AB})
    &\vcentcolon = \sum_{i,j,k,l=1,2,3} T_{i,j} T_{k,j} T_{k,l} T_{i,l}. \label{Eq:I4}
    \end{align}
\end{subequations}
Here, $T_{i,j} \vcentcolon = \tr(\rho_{AB} \sigma_i \otimes \sigma_j)$ for a two-qubit state $\rho_{AB}$, where $\sigma_i$ denotes the $i$-th Pauli matrix ($\sigma_0 = \eins_2$).

The LU invariants $I_t (\rho_{AB})$ for $t=2,4$ can capture the amount of two-body quantum correlations in the two-qubit state. For example, the product state $\ket{0} \otimes \ket{0}$ satisfies $I_t (\ket{0}\! \bra{0} \otimes \ket{0}\! \bra{0}) = 1$, while the Bell state $\ket{\Psi^-} = (1/\sqrt{2}) (\ket{0} \otimes \ket{1} - \ket{1} \otimes \ket{0})$ has $I_t (\ket{\Psi^-}\! \bra{\Psi^-}) = 3$. This indicates that entangled states exhibit stronger two-body correlations than product states.

More generally, a mixed state is called separable if it can be written as a convex combination of product states, $\rho_{\rm sep} = \sum_i p_i \ket{a_i}\! \bra{a_i} \otimes \ket{b_i}\! \bra{b_i}$, where $p_i$ represents a probability distribution. Otherwise, the state is entangled. In fact, it is known that any two-qubit separable state obeys $I_t (\rho_{AB}) \leq 1$ for $t=2,4$. Conversely, if $I_t (\rho_{AB}) > 1$, then $\rho_{AB}$ is entangled \cite{de_Vicente_2007}. This criterion directly follows from the three properties: (i) $I_t (\rho_{AB})$ are LU invariant; (ii) $I_t (\ket{0} \otimes \ket{0}) = 1$; (iii) $I_t (\rho_{AB})$ are convex for all quantum states, namely $I_t (\sum_i p_i \rho_{i}) \leq \sum_i p_i I_t (\rho_{i})$, for any ensemble $\{p_i,\rho_i \}$.

In the following, we represent $\rho_{AB}$ as the output of a quantum channel $\Lambda$ for a pure input state $\ket{\psi_{AB}}$:
\begin{equation} \label{eq:setting_rho_AB}
    \rho_{AB} = \Lambda (\ket{\psi_{AB}}).
\end{equation}
Based on $I_t (\rho_{AB})$ for $=2,4$, we will characterize quantum channels in terms of their ability to create, break, and preserve two-qubit entanglement. In Section~\ref{sec:ENTcreating}, we set the input state as a product state, i.e., $\ket{\psi_{AB}} = \ket{\psi_A} \otimes \ket{\psi_B}$, and take the channel $\Lambda$ acting on both systems $AB$ that can create entanglement. In Section~\ref{sec:ENTbreaking}, we choose the input state as the maximally entangled state, i.e., the reduced state for one particle of $\ket{\psi_{AB}}$ is maximally mixed, and take the channel $\Lambda$ acting only on the subsystem $B$ that can break and preserve the initial entanglement. In Section~\ref{sec:relation}, we compare these abilities induced from two-qubit unitary channels.

\section{Entanglement-creating channels}\label{sec:ENTcreating}
Let $\Lambda$ be a quantum channel on a two-particle system $AB$. The channel $\Lambda$ is said to be \emph{non-entangling} if it does not generate an entangled state for any input product state $\ket{\psi_A} \otimes \ket{\psi_B}$, i.e., $\Lambda (\ket{\psi_A} \otimes \ket{\psi_B})$ is always separable for all $\ket{\psi_A} \otimes \ket{\psi_B}$. Otherwise, it is called \emph{entangling}, i.e., $\Lambda (\ket{\psi_A} \otimes \ket{\psi_B})$ can be entangled for some $\ket{\psi_A} \otimes \ket{\psi_B}$ \cite{brylinski_2001,Bremner_2002}. Clearly, all local unitary channels are non-entangling, while there is only one non-local channel that is non-entangling, namely the SWAP operator $\swap_{AB}$, defined by $\swap_{AB} \ket{\psi_A} \otimes \ket{\psi_B} =\ket{\psi_B} \otimes \ket{\psi_A}$ \cite{Hulpke2006,Alfsen2010}. Hereafter, we will refer to an entangling channel as \emph{entanglement-creating} to distinguish it from \emph{entanglement-breaking} and \emph{entanglement-preserving} channels, which will be discussed in Section~\ref{sec:ENTbreaking}.

Our task here is to characterize entanglement-creating channels, but this is generally difficult for two main reasons: (i) deciding whether a channel is entanglement-creating depends on the specific choice of input product states; and (ii) determining whether the output (mixed) state is entangled requires the solution to a separability problem (see \cite{Guhne_2009}). Previous studies have addressed these issues by introducing input-independent quantities averaged over all possible product states and by focusing primarily on unitary channels \cite{zanardi2000}. For unitary channels, the analysis can be further simplified because the output state is pure and its entanglement can be completely determined only by the reduced density matrices on the subsystems. In contrast, for non-unitary channels, the output state is mixed, making its characterization much more involved than that of unitary channels.

To proceed, let us focus on a two-qubit channel $\Lambda$ and introduce input-independent quantities $\mathcal{C}^{(t)} (\Lambda)$ for $t=2,4$, referred to as \emph{moments}. These are expressed in terms of the LU invariants $I_t$, which directly capture two-body correlations rather than through reduced density matrices:
\begin{subequations}\label{Eq:C^(t)}
    \begin{align}
    \mathcal{C}^{(2)} (\Lambda)
    &\vcentcolon = \int d\psi_A \int d\psi_B \,
    I_2[\Lambda (\ket{\psi_A} \otimes \ket{\psi_B})],
    \\
    \mathcal{C}^{(4)} (\Lambda)
    &\vcentcolon = \int d\psi_A \int d\psi_B \,
    I_4[\Lambda (\ket{\psi_A} \otimes \ket{\psi_B})],
    \end{align}
\end{subequations}
where the states $\ket{\psi_X}$ for $X=A,B$ are typically chosen according to the Haar distribution with $\int d\psi_X = 1$ (see \cite{Mele2024}). By construction, $\mathcal{C}^{(t)} (\Lambda)$ for $t=2,4$ are independent of the choice of the input product states.

Now we can present the following:
\begin{observation}\label{ob:NE-channel-criterion}
    Any two-qubit non-entangling channel obeys $\mathcal{C}^{(t)} (\Lambda) \leq 1$ for $t=2,4$. Conversely, if $\mathcal{C}^{(t)} (\Lambda) > 1$, then the channel is entanglement-creating.
\end{observation}

\begin{proof}
    Recall the previously mentioned separability criterion: any two-qubit separable state obeys $I_t (\rho_{AB}) \leq 1$, where equality can be attained by pure states. Taking $\rho_{AB} = \Lambda (\ket{\psi_A} \otimes \ket{\psi_B})$ and using $\int d\psi_X = 1$, we can complete the proof.
\end{proof}

We make several remarks. First, the form of Eq.~(\ref{Eq:C^(t)}) can be simplified by considering the Kraus representation $\Lambda (\rho) = \sum_\alpha K_\alpha \rho K_\alpha^\dagger$ with $\sum_\alpha K_\alpha^\dagger K_\alpha = \eins$ for $K_\alpha$ being the Kraus operator acting on the system $AB$. By evaluating the Haar integrals, we have
\begin{widetext}
\begin{subequations}
    \begin{align}
    \mathcal{C}^{(2)} (\Lambda)
    &= \sum_{i,j=1}^3 \sum_{\alpha, \beta}
    \tr \! \left[
    K_{\alpha \beta}
    P_A^{(2)} \otimes P_B^{(2)}
    K_{\alpha \beta}^\dagger
    (\sigma_i^A \! \otimes \! \sigma_j^B \! \otimes \! 
    \sigma_i^A \! \otimes \! \sigma_j^B)
    \right],    
    \label{Eq:C^(2)}\\
    \mathcal{C}^{(4)} (\Lambda)
    &= \sum_{\substack{i,j,k,l=1}}^3
    \sum_{\alpha, \beta, \gamma, \delta}
    \tr \! \left[
    K_{\alpha \beta \gamma \delta}
    P_A^{(4)} \otimes P_B^{(4)}
    K_{\alpha \beta \gamma \delta}^\dagger
    (\sigma_i^A \! \otimes \! \sigma_j^B  \! \otimes \!
    \sigma_k^A \! \otimes \! \sigma_j^B \! \otimes \!
    \sigma_k^A \! \otimes \!  \sigma_l^B \! \otimes \!
    \sigma_i^A \! \otimes \! \sigma_l^B )
    \right],
    \label{Eq:C^(4)}
    \end{align}
\end{subequations}
\end{widetext}
where $K_{\alpha \beta} \vcentcolon = K_\alpha \otimes K_\beta$, $K_{\alpha \beta \gamma \delta} \vcentcolon = K_\alpha \otimes K_\beta \otimes K_\gamma \otimes K_\delta$, and $P_X^{(t)} = [1/(t+1)] \Tilde{P}_X^{(t)}$ for $X=A,B$ and $\Tilde{P}_X^{(t)}$ representing the projector onto the symmetric subspace of $t$ qubits for $t=2,4$. Here, $K_{\alpha \beta}$ acts on two copies of the two-qubit system, and $K_{\alpha \beta \gamma \delta}$ acts on four copies. The detailed calculation is provided in Appendix~\ref{ap:derivation_details_ENT_CRT}.

Next, for the channel $\Lambda (\rho) = \sum_\alpha K_\alpha \rho K_\alpha^\dagger$, it holds
\begin{equation}
     \label{eq:invariance}
     \mathcal{C}^{(t)}(\Lambda) = \mathcal{C}^{(t)}(\Lambda^{\prime}),
\end{equation}
where $\Lambda^{\prime} (\rho) = \sum_\alpha K_\alpha^{\prime} \rho (K_\alpha^{\prime})^\dagger$ with $K_\alpha^{\prime} = V_A \otimes V_B K_\alpha W_A \otimes W_B$ for local unitaries $V_A \otimes V_B$ and $W_A \otimes W_B$. This can be shown using two properties: (i) the LU invariance of $I_t$, i.e., $I_t(\rho_{AB}) = I_t (V_A \otimes V_B \rho_{AB} V_A^\dagger \otimes V_B^\dagger)$ for any local unitary $V_A \otimes V_B$; (ii) both left- and right-invariance of the Haar measure, i.e., $\int dU \, f(U) = \int dU \, f(UW) = \int dU \, f(WU)$ for any integrable function $f$ and for any unitary $W$, where we suppose that $\ket{\psi_A} \otimes \ket{\psi_B} = U_A \otimes U_B \ket{00}$.

In particular, for a unitary channel with $U_{AB}$, one can represent $\mathcal{C}^{(t)}(U_{AB}) = \mathcal{C}^{(t)}(U_{AB}^{\prime})$ with
\begin{equation} \label{eq:U_{AB}prime}
    U_{AB}^\prime = e^{i \sum_{i=1}^3 c_i \sigma_i \otimes \sigma_i},
\end{equation}
for $c_i \in [0, {\pi}/{4}]$. This follows the canonical decomposition introduced in Ref.~\cite{Zhang_2003}, which states that any two-qubit unitary $U_{AB}$ can be decomposed as $U_{AB} = V_A \otimes V_B U_{AB}^\prime W_A \otimes W_B$ for some unitaries $V_A, V_B, W_A, W_B$. By inserting Eq.~(\ref{eq:U_{AB}prime}) into Eqs.~(\ref{Eq:C^(2)},\ref{Eq:C^(4)}), we then conclude that the moments $\mathcal{C}^{(2)}(U_{AB})$ for any two-qubit unitary $U_{AB}$ can be given by
\begin{subequations}
    \begin{align}
    \mathcal{C}^{(2)}(U_{AB})
    &\!=\! \frac{1}{9} \Bigl\{ 15 - \cos(4c_{12}^-) - \cos(4c_{12}^+) \nonumber \\
    &\
    - 2 \bigl[\cos(4c_1) + \cos(4c_2)\bigr] \cos(4c_3)\Bigl\},
    \label{Eq:EC2}
    \\
    \mathcal{C}^{(4)} (U_{AB})
    &\!=\! \frac{1}{900} \Bigl\{ 1218+4 \cos(8 c_1)-76 \cos(4c_{12}^-) \nonumber \\
    &\
    +9 \cos(8c_{12}^-)+4 \cos(8 c_2) -76 \cos(4c_{12}^+) 
    \nonumber \\
    &\
    + 4 \bigl[ 
    1+3 \cos(4c_{12}^-) \bigl] \bigl[1 +3 \cos(4c_{12}^+) \bigl] \cos(8 c_3)    
    \nonumber\\
    &\
    +9 \cos(8c_{12}^+) + 8 \bigl[\cos(4 c_1)+\cos(4 c_2)\bigl]
    \nonumber \\
    &\
    \times \bigl[6 \cos(4 c_1) \cos(4 c_2) -22 \bigl] \cos(4 c_3)
    \Bigl\},
    \label{Eq:EC4}
    \end{align}
\end{subequations}
where $c_{12}^+ = c_1+c_2$ and $c_{12}^- = c_1-c_2$. The detailed calculation is provided in Appendix~\ref{ap:derivation_details_ENT_CRT_noisy_case}.

Third, for unitary channels, the second moment $\mathcal{C}^{(2)}(U_{AB})$ recovers the well-known expression of the \emph{entangling power} \cite{Reza2004}, up to some factors. This can be seen from the fact that, for any two-qubit state $\rho_{AB}$, the following identity holds:
\begin{equation}
    I_2(\rho_{AB}) = 1 + 4\tr(\rho_{AB}^2) - 2 [\tr(\rho_{A}^2) + \tr(\rho_{B}^2)],
\end{equation}
where $\rho_{X}$ for $X=A,B$ is the reduced density matrix of the subsystem $A,B$. In the case of unitary channels, $\rho_{AB}$ is a pure state, and then $\tr(\rho_{A}^2) = \tr(\rho_{B}^2)$ holds. Thus, the second moment $\mathcal{C}^{(2)}(\Lambda)$ can be regarded as a mixed-state generalization of the entangling power.

Fourth, the fourth moment $\mathcal{C}^{(4)}(\Lambda)$ can contain more precise information about the channel than the second moment $\mathcal{C}^{(2)}(\Lambda)$. In fact, $\mathcal{C}^{(4)}(\Lambda)$ can distinguish between certain channels that cannot be discriminated using only $\mathcal{C}^{(2)}(\Lambda)$. Examples are given by two unitary channels: the controlled-NOT (CNOT) unitary $U_{\rm CNOT} = e^{i \frac{\pi}{4} (\eins_2 - \sigma_3) \otimes (\eins_2-\sigma_1)}$ and the so-called B-unitary $U_{\rm B} = e^{i \frac{\pi}{8} (2 \sigma_1 \otimes \sigma_1 + \sigma_2 \otimes \sigma_2)}$. A straightforward calculation leads to
\begin{subequations}
\begin{align}
    \mathcal{C}^{(2)}(U_{\rm CNOT})
    &= \mathcal{C}^{(2)}(U_{\rm B}) = \frac{17}{9} \approx 1.889,
    \\
    \mathcal{C}^{(4)}(U_{\rm CNOT})
    &= \frac{353}{225} \approx 1.569,
    \\
    \mathcal{C}^{(4)}(U_{\rm B})
    &= \frac{23}{15} \approx 1.533.
\end{align}
\end{subequations}
Similar results have been reported by showing that higher-order moments can distinguish unitary channels with the same entangling power \cite{Cho_2026}.

On the other hand, the fourth moment $\mathcal{C}^{(4)}(\Lambda)$ cannot provide a full distinction of quantum channels, i.e., there exist some quantum channels that cannot be distinguished by $\mathcal{C}^{(4)}(\Lambda)$. Examples are the CNOT unitary $U_{\rm CNOT}$ and the so-called $i$SWAP unitary $U_{i {\rm SWAP}} = R_{xx}(- {\pi}/{2})R_{yy}(- {\pi}/{2})$ with $R_{\alpha \alpha}(\theta) = e^{- i \frac{\theta}{2} \sigma_\alpha \otimes \sigma_\alpha}$. Note that previous studies have addressed this issue by distinguishing between local and non-local two-qubit unitaries using \emph{gate-typicality} \cite{Jonnadula2017}.

Finally, we note that Observation~\ref{ob:NE-channel-criterion} is not a necessary condition for two-qubit entanglement-creating channels. Then, there may be some entanglement-creating channels that satisfy the bound $\mathcal{C}^{(t)} (\Lambda) \leq 1$. Our next aim is to characterize such channels by finding another criterion that can be stronger than Observation~\ref{ob:NE-channel-criterion}.

One approach is to use a combination of $\mathcal{C}^{(2)}(\Lambda)$ and $\mathcal{C}^{(4)}(\Lambda)$, motivated by previous studies in entanglement detection via randomized measurements \cite{Ketterer_2019,Ketterer2020,Knips_2020, Ketterer_2022,Wyderka2023,CIESLINSKI2024}. In contrast, deriving criteria for an \textit{arbitrary} two-qubit entanglement-creating channel is generally more challenging than establishing entanglement criteria for states. This is because channel characterization requires specifying the full dynamical map rather than a single output state, leading to a higher-dimensional parameter space and additional constraints such as complete positivity and trace preservation. 

To proceed, let us fix a concrete channel model:
\begin{equation}
    \Lambda_p (\rho) = p U_{AB} \rho U_{AB}^\dagger + (1-p)\frac{\eins_{AB}}{4},
    \label{eq:Lambda_p}
\end{equation}
which represents a probabilistic mixture of a two-qubit unitary channel and complete depolarization as a noisy effect. We can then summarize our results as follows:
\begin{observation}   \label{ob:C2-C4-NE-channel-criterion}
    Consider the space spanned by the moments $\mathcal{C}^{(2)}(\Lambda_p)$ and $\mathcal{C}^{(4)}(\Lambda_p)$ for the channel $\Lambda_p$ as in Eq.~(\ref{eq:Lambda_p}), shown in Fig.~\ref{fig:EC_plot}~(a). For any two-qubit unitary $U_{AB}$, if the channel $\Lambda_p$ is non-entangling, then it obeys $F_{\rm NE}[\mathcal{C}^{(2)}(\Lambda_p), \mathcal{C}^{(4)}(\Lambda_p)] \geq 0$, where $F_{\rm NE}$ defines the boundary of the non-entangling region. A violation of this inequality, corresponding to a point lying outside the non-entangling region, identifies a broader class of entanglement-creating cases, compared to the criterion based solely on $\mathcal{C}^{(t)}(\Lambda_p) \leq 1$ given in Observation~\ref{ob:NE-channel-criterion}.
\end{observation}  

\begin{figure}[t]
        \begin{overpic}[width=1.0\linewidth]{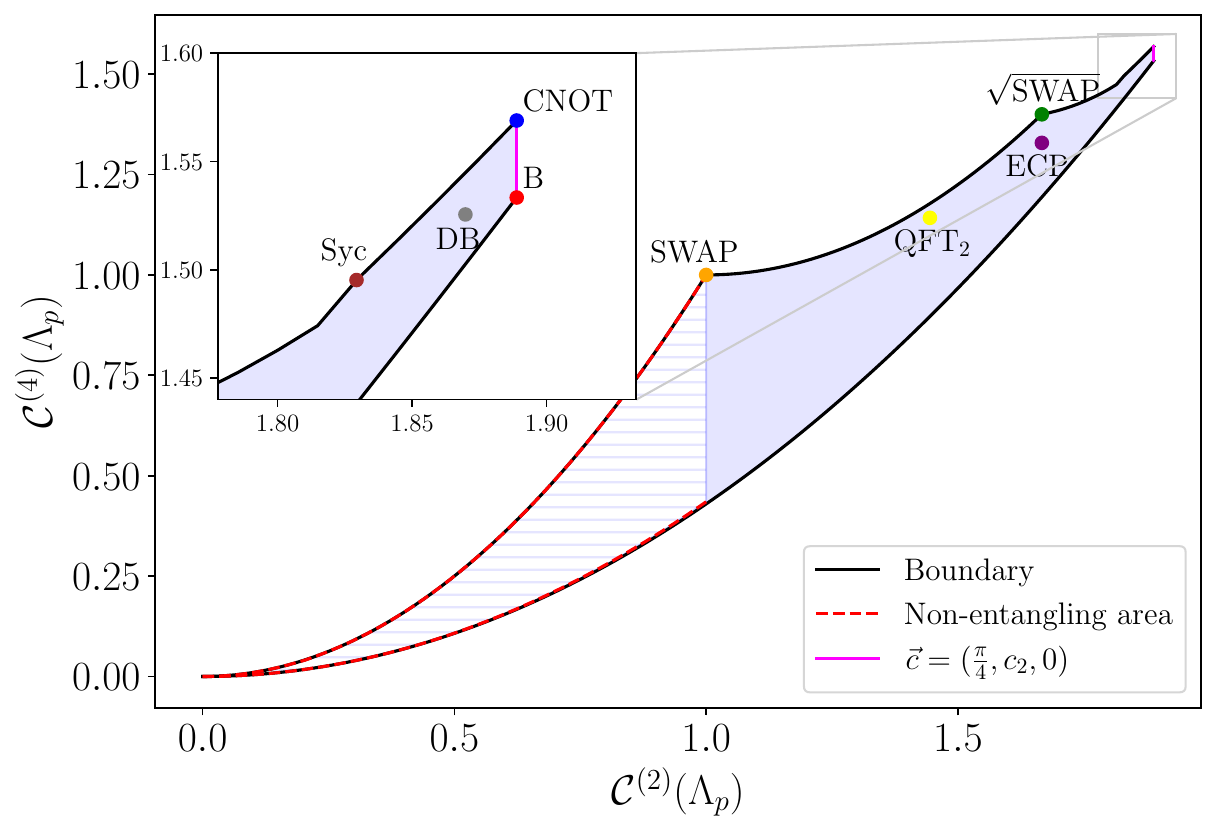}
            \put(2,72){{(a)}}  
        \end{overpic}
    
        \vspace{0.3cm}
    
        \begin{overpic}[width=1.0\linewidth]{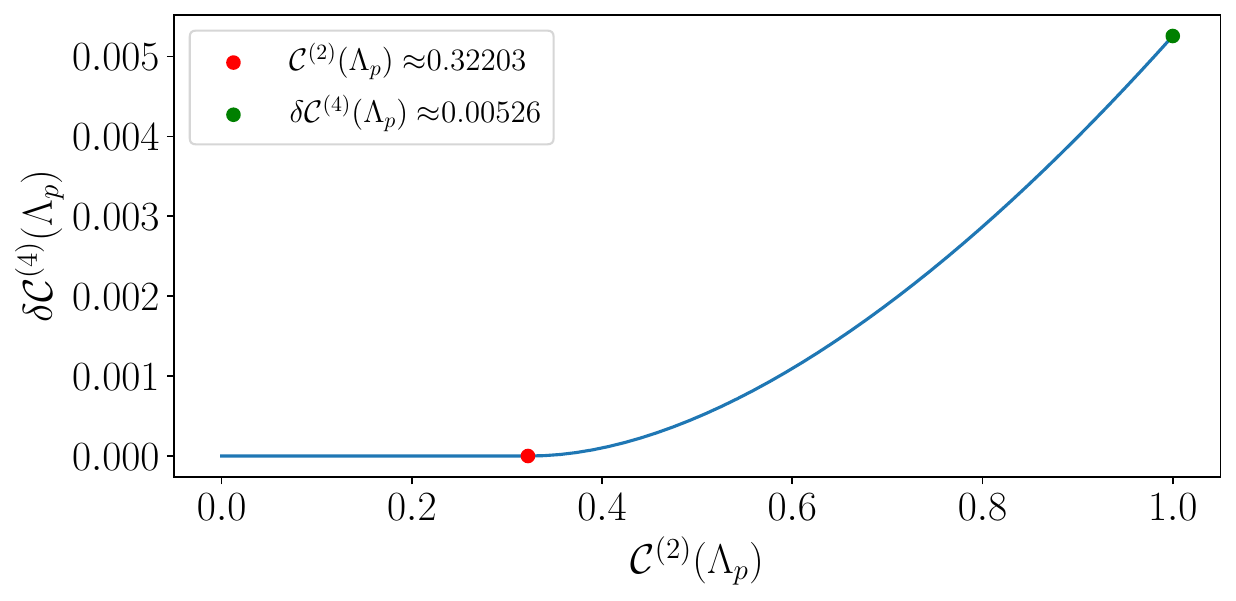}
            \put(2,52){{(b)}}
        \end{overpic}
        
        \caption{(a) Representation of entangling and non-entangling channels in the space spanned by the second and fourth moments, $\mathcal{C}^{(2)}(\Lambda_p)$ and $\mathcal{C}^{(4)}(\Lambda_p)$, for a two-qubit channel model $\Lambda_p$ defined in Eq.~(\ref{eq:Lambda_p}). All non-entangling cases are contained in the hashed region characterized by the dashed red curves with $F_{\rm NE}[\mathcal{C}^{(2)}(\Lambda_p), \mathcal{C}^{(4)}(\Lambda_p)] \geq 0$, presented in Observation~\ref{ob:C2-C4-NE-channel-criterion}. The blue region corresponds to entanglement-creating cases, which marks the improvement of the criterion in Observation~\ref{ob:NE-channel-criterion}. The inset shows several well-known unitaries with $p=1$ according to their canonical parameters $\vec{c}$ given in Table \ref{tab:gates}. (b) The gap $\delta \mathcal{C}^{(4)}(\Lambda_p)$ between the red and black lower curves in Fig.~1 (a) with increasing $\mathcal{C}^{(2)}(\Lambda_p)$.}
        \label{fig:difference_with_ppt}
        \label{fig:EC_plot}
\end{figure}

\begin{table*}[t]
    \centering
    \begin{adjustbox}{width=0.75\textwidth}
    \begin{tabular}{|c|c c c c c|}
        \hline
          \backslashbox{$U_{AB}$ }{} & $ \vec{c} $ & \hspace{0.35cm} $\mathcal{C}^{(2)}(U_{AB}) $ \hspace{0.4cm} &  $ \mathcal{C}^{(4)}(U_{AB}) $ & \hspace{0.4cm}$\mathcal{P}^{(2)}(\Gamma_{U_{AB}})$ & \hspace{0.4cm} $\mathcal{P}^{(4)}(\Gamma_{U_{AB}})$ \\
        \hline
        CNOT & $(\frac{\pi}{4}, 0, 0)$ &  17/9  &  353/225 
         & 1 & 1\\
        B & $(\frac{\pi}{4}, \frac{\pi}{8}, 0)$ & 17/9  &   23/15 & 1/2 & 1/4\\
        $i$SWAP & $(\frac{\pi}{4}, \frac{\pi}{4}, 0)$ &   17/9  &   353/225  & 0 & 0  \\
        DB & $(\frac{3\pi}{16}, \frac{3\pi}{16}, 0)$ & $\approx$1.870 & $\approx$1.526 & $\approx 0.3143$ & $\approx 0.0433$\\
        Syc & $(\frac{\pi}{4}, \frac{\pi}{4}, \frac{\pi}{24})$ & $\approx$1.829 & $\approx$1.495& 0 & 0 \\
        $\sqrt{\text{SWAP}}$ & $(\frac{\pi}{8}, \frac{\pi}{8}, \frac{\pi}{8})$ & 15/9  & 7/5  & 3/4 & 3/16\\
        ECP & $(\frac{\pi}{4}, \frac{\pi}{8}, \frac{\pi}{8})$ & 15/9  & 299/225  & 1/4 & 1/16\\
        QFT$_2$ & $(\frac{\pi}{4}, \frac{\pi}{4}, \frac{\pi}{8})$ & 13/9  & 257/225  & 0 & 0\\
       SWAP & $(\frac{\pi}{4}, \frac{\pi}{4}, \frac{\pi}{4})$ & 1 & 1 & 0 & 0\\ \hline
    \end{tabular}
    \end{adjustbox}
    \caption{Common two-qubit unitaries' canonical parameters $\vec{c}=(c_1, c_2, c_3)$ and their values of $\mathcal{C}^{(2)}(U_{AB})$,  $\mathcal{C}^{(4)}(U_{AB})$, $\mathcal{P}^{(2)}(\Gamma_{U_{AB}})$, and $\mathcal{P}^{(4)}(\Gamma_{U_{AB}})$ illustrated in Fig.~\ref{fig:EC_plot} and Fig.~\ref{fig:Pt_Ct}, respectively. The CNOT, $i$SWAP, $\sqrt{\mathrm{SWAP}}$, and SWAP gates are well-known two-qubit unitaries. The B-gate and both DB and ECP gates are two-qubit gates introduced in \cite{Zhang_2004} and \cite{Peterson_2020}, respectively. The Syc-gate is the native two-qubit gate on Google's Sycamore quantum computing architecture \cite{Google_2019}. The QFT$_2$-gate refers to the quantum Fourier transform two-qubit gate.
    }
    \label{tab:gates}
\end{table*}

\begin{proof}
    We begin by writing
    \begin{equation}
        I_t[\Lambda_p (\ket{\psi_A} \otimes \ket{\psi_B})]
        = p^t I_t[ U_{AB} \ket{\psi_A} \otimes \ket{\psi_B}],
    \end{equation}
    where we used $T_{i,j} [\Lambda_p (\ket{\psi_A} \otimes \ket{\psi_B} )] = p T_{i,j} [U_{AB}\ket{\psi_A} \otimes \ket{\psi_B}]$ for $T_{i,j} [\rho_{AB}] = \tr(\rho_{AB} \sigma_i \otimes \sigma_j)$. Then, we have
    \begin{equation}
        \mathcal{C}^{(t)}(\Lambda_p)
        = p^t \mathcal{C}^{(t)}(U_{AB}),
    \end{equation}
    where $\mathcal{C}^{(t)}(U_{AB})$ for $t=2,4$ have been evaluated previously in Eqs.~(\ref{Eq:EC2},\ref{Eq:EC4}). Consequently, we can plot Fig.~\ref{fig:EC_plot} that represents the space spanned by the moments $\mathcal{C}^{(2)}(\Lambda_p)$ and $\mathcal{C}^{(4)}(\Lambda_p)$ for a general $\Lambda_p$. This region can be obtained by maximizing/minimizing $\mathcal{C}^{(4)}(\Lambda_p)$ over $p, c_1, c_2, c_3$ for a fixed value $\mathcal{C}^{(2)}(\Lambda_p)$, subject to $0 \leq p \leq 1$ and $0\leq c_i \leq \pi/4$.

    To identify the border of the non-entangling set, we must impose additional constraints on the above optimization. In fact, we can show that such constraints are given by the separability condition of the state
    \begin{equation}\label{eq:rhoprimesep}
    \rho_{AB}^\prime
    = p U_{AB}^\prime \ket{00} \! \bra{00} (U_{AB}^\prime)^\dagger
    + (1-p)\frac{\eins_{AB}}{4},
    \end{equation}
    where $U_{AB}^\prime$ is given in Eq.~(\ref{eq:U_{AB}prime}).
    
    This can be shown as follows. First, recall that $\Lambda_p$ is non-entangling if $\Lambda_p (\ket{\psi_A} \otimes \ket{\psi_B})$ is separable for any input product state. Also, separability is invariant under local unitary transformations: for any separable state $\rho_{\rm sep}$ and for any local unitary $V_A \otimes V_B$, the state $V_A \otimes V_B \rho_{\rm sep} V_A^\dagger \otimes V_B^\dagger$ is also separable. By letting $\ket{\psi_X} = W_X \ket{0}$ for $X=A,B$, we find that $\Lambda_p$ is non-entangling if the state $\rho_{AB}^\prime$ as in Eq.~(\ref{eq:rhoprimesep}) is separable, where $U_{AB}^\prime = V_A \otimes V_B U_{AB} W_A \otimes W_B$.
    
    For two-qubit systems, the separability problem is solved by the positive partial transpose (PPT) criterion \cite{Peres_1996, Horodecki_1996}: Any two-qubit state is separable if and only if its partial transpose state is positive semidefinite. Thus, applying the PPT criterion for the state $\rho_{AB}^\prime$ gives further parameter constraints on $p, c_1, c_2, c_3$. Incorporating the resulting constraints into the optimization determines the boundary of the non-entangling region corresponding to the separable domain.
\end{proof}

In Fig.~\ref{fig:EC_plot}~(a), several well-known two-qubit unitary gates with $p=1$ are located according to their canonical parameters $\vec{c} = (c_1 , c_2, c_3)$, as listed in Table~\ref{tab:gates}. All listed unitaries, except SWAP and QFT$_2$, belong to the class of \emph{perfect entanglers}, which can generate a maximally entangled state from a product input state. Among them, the CNOT, B, and $i$SWAP gates are called the \emph{special perfect entanglers}, as they can transform an entire orthonormal product basis into a maximally entangled basis \cite{Reza2004}. All such unitaries correspond to $\vec{c} = (\pi/4, c_2, 0)$ for $c_2 \in [0, \pi/4]$ and lie along the pink vertical line with $\mathcal{C}^{(2)}(U_{AB}) = 17/9$.

Also, the solid black curves in Fig.~\ref{fig:EC_plot}~(a) represent the whole bounds for the moments of $\Lambda_p$, while the dashed red curves indicate the non-entangling region. Here, the dashed red upper bound coincides exactly with the solid black upper bound, which corresponds to the case where the unitary $U_{AB}$ in Eq.~(\ref{eq:Lambda_p}) is the SWAP operator $\swap_{AB}$ for all values of $p$. In contrast, the dashed red lower and solid black lower bounds do not match, although this is invisible in Fig.~\ref{fig:EC_plot}~(a). This deviation is illustrated in Fig.~\ref{fig:difference_with_ppt}~(b), which shows the difference $\delta \mathcal{C}^{(4)}(\Lambda_p)$ between both lower bounds as a function of $\mathcal{C}^{(2)}(\Lambda_p)$. There the gap becomes noticeable around $\mathcal{C}^{(2)}(\Lambda_p) \approx 0.32203$ (red point) and reaches its maximum value of $\delta\mathcal{C}^{(4)}(\Lambda_p) \approx 0.00526$ (green point) at $\mathcal{C}^{(2)}(\Lambda_p) = 1$. The optimization was performed with tight convergence tolerances of $10^{-12}$ (details are in Appendix~\ref{ap:optimization_details}). 

More analytically, the solid black lower curve corresponds to the case with $\vec{c}=(\pi/4, \pi/8, 0)$ for all values of $p$. The (invisible) gap region between the dashed red lower and solid black lower bounds contains an entanglement-creating class with $\vec{c}=(\pi/4, c_2, 0)$, where $c_2 \in [\pi/8,\pi/4]$ and $p \in (1/3, 1]$. In fact, we can show that $p=1/3$ is the exact transition point between the entanglement-creating and non-entangling regimes for this family, meaning that all cases with $\vec{c}=(\pi/4, c_2, 0)$ for $p \leq 1/3$ are non-entangling. This follows from the fact that Eq.~(\ref{eq:Lambda_p}) becomes LU equivalent to the two-qubit Werner state \cite{Werner_1989} at $\vec{c}=(\pi/4, c_2, 0)$, since the unitary $U_{AB}$ in Eq.~(\ref{eq:Lambda_p}) belongs to the class of special perfect entanglers.

\section{Entanglement-breaking and entanglement-preserving channels} \label{sec:ENTbreaking}
In the previous section, we have characterized the ability to create entanglement from a product state using the moments $\mathcal{C}^{(t)}(\Lambda)$ for a quantum channel $\Lambda$ acting on the two-particle system $AB$. Below, we instead analyze the abilities to break and preserve entanglement in an initial state for a channel acting only on one subsystem: $\Lambda = {\rm id} \otimes \Gamma$. Here, ${\rm id}$ is the identity channel on the subsystem $A$, and $\Gamma$ is a channel on the subsystem $B$.

Specifically, the channel $\Gamma$ is called \emph{entanglement-breaking} if it generates a separable state for any input entangled state $\rho_{\rm ent}$, i.e., $({\rm id} \otimes \Gamma) (\rho_{\rm ent})$ is always separable for all $\rho_{\rm ent}$ \cite{Horodecki_2003}. Otherwise, we call it \emph{entanglement-preserving}, i.e., $({\rm id} \otimes \Gamma) (\rho_{\rm ent})$ can be entangled for some $\rho_{\rm ent}$. It is known that a channel $\Gamma$ is entanglement-breaking if and only if $({\rm id} \otimes \Gamma) (\ket{\Omega}\! \bra{\Omega})$ is separable \cite{Ruskai_2003}, where $\ket{\Omega} = (1/\sqrt{d}) \sum_{i} \ket{i}\otimes \ket{i}$ is a maximally entangled state.

In the following, we focus on a single-qubit channel $\Gamma$ and characterize its entanglement-breaking and entanglement-preserving properties. To this end, similarly to Eq.~(\ref{Eq:C^(t)}), we introduce the moments $\mathcal{P}^{(t)} (\Gamma)$ for $t=2,4$ defined through the LU invariants $I_t$:
\begin{subequations}\label{Eq:B^(t)}
    \begin{align}
    \mathcal{P}^{(2)} (\Gamma)
    &\vcentcolon = \int dU \,
    I_2[({\rm id} \otimes \Gamma) (\ket{\Omega_{U}}\! \bra{\Omega_{U}})],
    \\
    \mathcal{P}^{(4)} (\Gamma)
    &\vcentcolon = \int dU \,
    I_4[({\rm id} \otimes \Gamma) (\ket{\Omega_{U}}\! \bra{\Omega_{U}})],
    \end{align}
\end{subequations}
where $\ket{\Omega_{U}} = (\eins \otimes U) \ket{\Omega}$ and the unitary $U$ is chosen according to the Haar distribution with $\int dU= 1$. Here, it is not necessary to apply a unitary on subsystem $A$ because of the LU invariance of $I_t$. We remark that, without loss of generality, $\ket{\Omega_{U}}$ in Eq.~(\ref{Eq:B^(t)}) can be replaced with $\ket{\Psi^-_{U}} = (\eins \otimes U) \ket{\Psi^-}$ where $\ket{\Psi^-} = (1/\sqrt{2}) (\ket{0} \otimes \ket{1} - \ket{1} \otimes \ket{0})$.

Now we can make the following:
\begin{observation}\label{ob:EB-channel}
    Any single-qubit entanglement-breaking channel obeys $\mathcal{P}^{(t)} (\Gamma) \leq 1$ for $t=2,4$. Conversely, if $\mathcal{P}^{(t)} (\Gamma) > 1$, then the channel is entanglement-preserving.
\end{observation}

\begin{proof}
    Similarly to the proof of Observation~\ref{ob:NE-channel-criterion}, we can complete the proof.
\end{proof}

We have several remarks. First, we can simplify Eq.~(\ref{Eq:B^(t)}) using the Kraus representation $\Gamma (\rho) = \sum_\alpha E_\alpha \rho E_\alpha^\dagger$ with $\sum_\alpha E_\alpha^\dagger E_\alpha = \eins$, where $E_\alpha$ is the Kraus operator acting on the system $B$. By denoting $\Phi_t (\mathcal{O})$ for an Hermitian operator $\mathcal{O}$ in $t$ qubits as
\begin{equation}
    \Phi_t (\mathcal{O}) \vcentcolon = \int dU \, U^{\otimes t} \mathcal{O} (U^\dagger)^{\otimes t},
\end{equation}
we obtain
\begin{widetext}
\begin{subequations}
    \begin{align}
        \mathcal{P}^{(2)} (\Gamma)
        &= \frac{1}{4} \sum_{i,j} \sum_{\alpha , \beta}
        \tr \left[  E_{\alpha \beta} \Phi_2(\sigma_i \otimes \sigma_i) E^\dagger_{\alpha \beta} (\sigma_j \otimes \sigma_j) \right], 
        \label{Eq:E^(2)}
        \\
        \mathcal{P}^{(4)} (\Gamma)
        &= \frac{1}{16} \sum_{i,j, k,l} \sum_{\alpha, \beta, \gamma, \delta}
        \tr \left[  E_{\alpha \beta \gamma \delta} \Phi_4(\sigma_i \otimes \sigma_k \otimes \sigma_k \otimes \sigma_i) E^\dagger_{\alpha \beta \gamma \delta} (\sigma_j \otimes \sigma_j \otimes \sigma_l \otimes \sigma_l) \right],  
        \label{Eq:E^(4)}
    \end{align}
\end{subequations}
\end{widetext}
where $E_{\alpha \beta} \vcentcolon = E_\alpha \otimes E_\beta$ and $E_{\alpha \beta \gamma \delta} \vcentcolon = E_\alpha \otimes E_\beta \otimes E_\gamma \otimes E_\delta$. The detailed calculation is provided in Appendix~\ref{ap:derivation_E^(2)andE^(4)}.

Second, for the identity channel $\Gamma = {\rm id}$, i.e., $\Gamma (\rho) = \rho$, we have that $\mathcal{P}^{(t)}({\rm id}) = I_t[\ket{\Psi^-}] = 3$. Also, for the unitary channel $\Gamma_U (\rho)= U \rho U^\dagger$, we have $\mathcal{P}^{(t)}(\Gamma_U) = 3$. For the completely depolarizing channel $\Gamma_{\rm CD} (\rho) = \eins/2$ for any $\rho$, we obtain $\mathcal{P}^{(t)}(\Gamma_{\rm CD}) = I_t[(1/4) \eins \otimes \eins] = 0$. 

Third, Observation~\ref{ob:EB-channel} is not a necessary condition for entanglement-preserving channels. Thus, there are some entanglement-preserving channels that satisfy the bound $\mathcal{P}^{(t)} (\Gamma) \leq 1$. Below, we will characterize such channels by developing additional criteria beyond Observation~\ref{ob:EB-channel}.

Our approach is to again consider a combination of $\mathcal{P}^{(2)}(\Gamma)$ and $\mathcal{P}^{(4)}(\Gamma)$ in a similar manner to the previous section. To proceed, we focus on the \textit{unital} (often called \textit{bistochastic}) channel that satisfies $\Gamma_{\rm U} (\eins) = \eins$, meaning that $\eins$ is a fixed point. The unital channel satisfies $\sum_\alpha E_\alpha E_\alpha^\dagger = \eins$ in addition to the trace-preserving condition $\sum_\alpha E_\alpha^\dagger E_\alpha = \eins$. It is known that any single-qubit unital channel can be represented as a convex combination of unitaries \cite{Bengtsson_Zyczkowski_2006}:
\begin{equation} \label{eq:single-qubit_unital}
    \Gamma_{\rm U} (\rho) = \sum_\alpha p_\alpha U_\alpha \rho U_\alpha^\dagger,
\end{equation}
where $U_\alpha$ is a unitary operator and $p_\alpha$ represents a probability distribution: $E_\alpha = \sqrt{p_\alpha} U_\alpha$. This suggests that the unital channel can describe a broad class of quantum channels.

\begin{figure}[t]
        \begin{overpic}[width=1.0\linewidth]{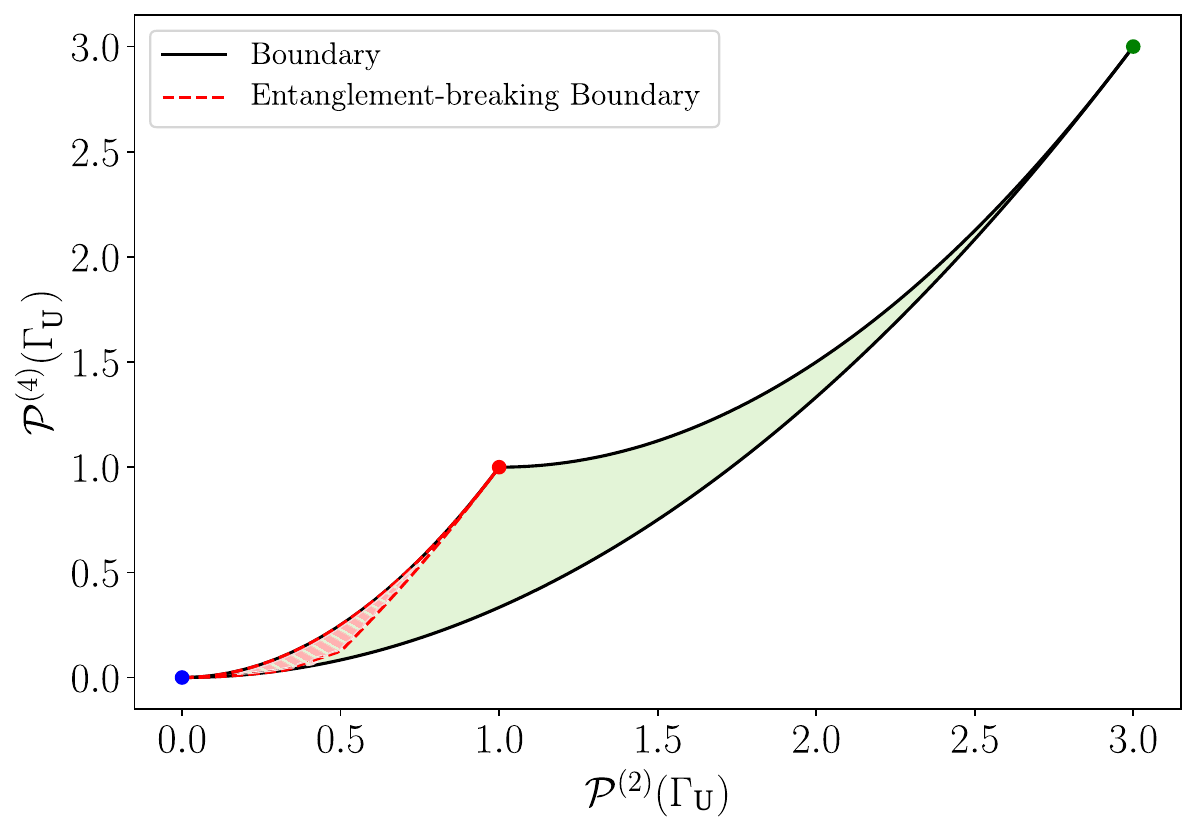}
            \put(2,72){{(a)}}  
        \end{overpic}
    
        \vspace{0.3cm}
    
        \begin{overpic}[width=1.02\linewidth]{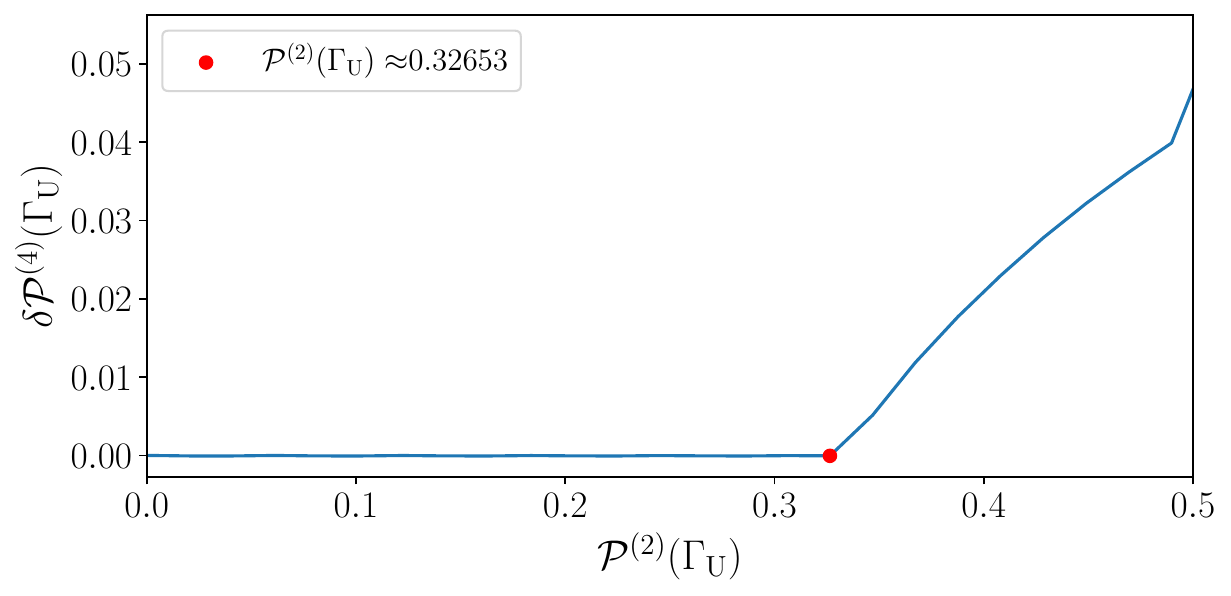}
            \put(2,52){{(b)}}
        \end{overpic}
        
        \caption{(a) Representation of entanglement-preserving and entanglement-breaking channels in the space spanned by the second and fourth moments, $\mathcal{P}^{(2)}(\Gamma_{\rm U})$ and $\mathcal{P}^{(4)}(\Gamma_{\rm U})$, for a single-qubit unital channel $\Gamma_{\rm U}$. All entanglement-breaking channels are contained in the red hashed region with $F_{\rm EB}[\mathcal{P}^{(2)}(\Gamma_{\rm U}), \mathcal{P}^{(4)}(\Gamma_{\rm U})] \geq 0$, presented in Observation~\ref{ob:B2-B4-EB-channel}. The green region corresponds to entanglement-preserving cases, which marks the improvement of the criterion in Observation~\ref{ob:EB-channel}. (b) The gap $\delta \mathcal{P}^{(4)}(\Gamma_{\rm U})$ between the dashed red and solid black lower curves in Fig.~\ref{fig:EB_plot}~(a) with increasing $\mathcal{P}^{(2)}(\Gamma_{\rm U})$. }
        \label{fig:difference_with_EB}
        \label{fig:EB_plot}
\end{figure}

We can summarize our results:
\begin{observation}   \label{ob:B2-B4-EB-channel}
    Consider the space spanned by the moments $\mathcal{P}^{(2)}(\Gamma_{\rm U})$ and $\mathcal{P}^{(4)}(\Gamma_{\rm U})$ for the unital channel $\Gamma_{\rm U}$ as in Eq.~(\ref{eq:single-qubit_unital}), shown in Fig.~\ref{fig:EB_plot}. For any single-qubit unital channel $\Gamma_{\rm U}$, if $\Gamma_{\rm U}$ is entanglement-breaking, then it obeys $F_{\rm EB}[\mathcal{P}^{(2)}(\Gamma_{\rm U}), \mathcal{P}^{(4)}(\Gamma_{\rm U})] \geq 0$, where $F_{\rm EB}$ defines the boundary of the entanglement-breaking region. A violation of this inequality, corresponding to a point lying outside the entanglement-breaking region, identifies a broader class of entanglement-preserving cases, compared to the criterion based solely on $\mathcal{P}^{(t)}(\Gamma_{\rm U}) \leq 1$ given in Observation~\ref{ob:EB-channel}.
\end{observation}  

\begin{proof}
    We begin by noting that for a quantum channel $\Gamma (\rho) = \sum_\alpha E_\alpha \rho E_\alpha^\dagger$, it holds that
    \begin{equation} \label{eq:invariance_B}
        \mathcal{P}^{(t)} (\Gamma)
        = \mathcal{P}^{(t)} (\Gamma^\prime),
    \end{equation}
    where $\Gamma^\prime (\rho) = \sum_\alpha E_\alpha^{\prime} \rho (E_\alpha^{\prime})^\dagger$ with $E_\alpha^{\prime} = V E_\alpha W$ for unitaries $V,W$. This can be shown similarly with Eq.~(\ref{eq:invariance}), by recalling the LU invariance of $I_t$ and both left- and right-invariance of the Haar measure. Note that two channels $\Gamma$ and $\Gamma^\prime$ are called \textit{unitarily similar} (or \textit{unitarily equivalent}) if they satisfy the relation $\Gamma^\prime (\rho) = V \Gamma (W \rho W^\dagger) V^\dagger$. Thus, the moments $\mathcal{P}^{(t)} (\Gamma)$ can characterize a unitarily similar class of a quantum channel $\Gamma$.

    Next, we note that any single-qubit unital channel in a unitarily similar class can be written as 
    \begin{equation}\label{eq:unitalform_Pauli}
        \Gamma_{\rm U} (\rho) = 
        \sum_{\alpha = 0}^3 d_\alpha \sigma_\alpha \rho \sigma_\alpha,
        \quad
        \sum_{\alpha = 0}^3 d_\alpha = 1,
    \end{equation}
    where $\sigma_\alpha$ denotes the $\alpha$-th Pauli matrix ($\sigma_0 = \eins_2$) and $d_\alpha \in [0,1]$ \cite{Bengtsson_Zyczkowski_2006}. Thus, without loss of generality, we can take the form in Eq.~(\ref{eq:unitalform_Pauli}) to simplify the moments $\mathcal{P}^{(t)} (\Gamma_{\rm U})$. Since the Kraus operator is in the form of $E_\alpha = \sqrt{d_\alpha} \sigma_\alpha$, inserting this into Eqs.~(\ref{Eq:E^(2)},\ref{Eq:E^(4)}) leads to
    \begin{subequations}
        \begin{align}
            \mathcal{P}^{(2)} (\Gamma_{\rm U})
            &= \sum_{\alpha=1}^3  D_\alpha^2, 
            \label{eq:B2_unital_form}
            \\
            \mathcal{P}^{(4)} (\Gamma_{\rm U})
            &= \sum_{\alpha=1}^3 D_\alpha^4,
            \label{eq:B4_unital_form}
        \end{align}
    \end{subequations}
    where $D_\alpha = 2( d_\alpha +  d_0) - 1$. The detailed calculation is provided in Appendix~\ref{ap:derivation_B2andB4_unital_form}.

    Consequently, we consider the space spanned by the moments $\mathcal{P}^{(2)}(\Gamma_{\rm U})$ and $\mathcal{P}^{(4)}(\Gamma_{\rm U})$ for any single-qubit unital channel $\Gamma_{\rm U}$. We can obtain this region by maximizing/minimizing $\mathcal{P}^{(4)}(\Gamma_{\rm U})$ over $d_0, d_1, d_2, d_3$ for a fixed value $\mathcal{P}^{(2)}(\Gamma_{\rm U})$, subject to $d_\alpha \in [0,1]$ and $\sum_{\alpha = 0}^3 d_\alpha = 1$.

    Let us impose further constraints on the above optimization to characterize the border of the entanglement-breaking set. We can show that such constraints are
    \begin{equation}
        \sum_{\alpha=1}^3 |D_\alpha| \leq 1,
        \label{eq:EB_constraints}
    \end{equation}
    which is derived below. Hence, by adding the constraint in Eq.~(\ref{eq:EB_constraints}) to the above optimization, we can identify the entanglement-breaking boundary.

    Finally, let us derive Eq.~(\ref{eq:EB_constraints}). First we recall the Bloch representation of the single-qubit state $\rho$: $\rho = (1/2) ( \eins + \sum_{i=1}^3 r_i \sigma_i )$. Inserting this to $\Gamma_{\rm U} (\rho)$ in Eq.~(\ref{eq:unitalform_Pauli}) and using the formula $\sigma_\alpha \sigma_i \sigma_\alpha = (1-\delta_{\alpha 0})(2\delta_{\alpha i}\sigma_\alpha - \sigma_i) + \delta_{\alpha 0} \sigma_\alpha$ for $\alpha=0,1,2,3$ and $i=1,2,3$, we obtain that $\Gamma_{\rm U}(\rho) = (1/2) ( \eins + \sum_{\alpha=1}^3 D_\alpha r_\alpha \sigma_i)$, where $D_\alpha = 2( d_\alpha +  d_0) - 1$. This indicates that the unital channel $\Gamma_{\rm U}$ affects the Bloch vector components $r_\alpha$ by simply scaling them as $D_\alpha r_\alpha$. It is known that any single-qubit unital channel is entanglement-breaking if and only if $\sum_{\alpha=1}^3 \left| D_\alpha \right| \leq 1$ \cite{Ruskai_2003}. Thus, we can arrive at Eq.~(\ref{eq:EB_constraints}).
\end{proof}

In Fig.~\ref{fig:EB_plot}~(a), the solid black curves represent the whole bounds for the moments of $\Gamma_{\rm U}$, while the dashed red curves indicate the entanglement-breaking region. Here, the dashed red upper bound coincides exactly with the solid black upper bound, which corresponds to the case where the parameters $\vec{d} = (d_0, d_1, d_2, d_3)$ in Eq.~\eqref{eq:unitalform_Pauli} are given by $(1/2-x, 1/2-x, x, x)$, $(1/2-x, x, 1/2-x, x)$, and $(1/2-x, x, x, 1/2-x)$ for $x \in [0,1/4]$. In particular, the red point at $\mathcal{P}^{(t)}(\Gamma_{\rm U}) = 1$ for $t=2,4$ corresponds to the case with $x=0$ in the above, the green point at $\mathcal{P}^{(t)}(\Gamma_{\rm U}) = 3$ to the identity and unitary channels with $\vec{d} = (1,0,0,0)$, and the blue point at $\mathcal{P}^{(t)}(\Gamma_{\rm U}) = 0$ to the completely depolarizing channel $\vec{d} = (1/4,1/4,1/4,1/4)$.

In contrast, the dashed red lower and solid black lower bounds do not match, whose deviation is shown in Fig.~\ref{fig:EB_plot}~(b). There, the gap $\delta \mathcal{P}^{(4)}(\Lambda_p)$ among them becomes noticeable around $\mathcal{P}^{(2)}(\Gamma_{\rm U}) \approx 0.32653$ (red point), where the optimization was performed with tight convergence tolerances of $10^{-12}$ (details are in Appendix~\ref{ap:optimization_details}). More analytically, the solid black lower curve in Fig.~\ref{fig:EB_plot}~(a) corresponds to $\vec{d} = (1-3d,d,d,d)$ for $d \in [0,1/4]$, while the part of the dashed red curve where it no longer overlaps with the solid black lower curve corresponds to $\vec{d}=(1/2-2x,x,x,1/2)$ for $x \in [0,1/4]$. This leads to the gap starting to appear at exactly $\mathcal{P}^{(2)} (\Gamma_{\rm U}) = 1/3$ for $x = 1/6$.

\section{Relation between two-qubit unitary and single-qubit noisy channel} \label{sec:relation}
In this section, we consider a scenario in which a single-qubit state $\zeta$ interacts with a maximally mixed environment state $\eins/2$ through a two-qubit unitary $U_{AB}$. The induced dynamics can be described by the single-qubit channel
\begin{equation} \label{eq:noisy-unitary}
    \Gamma_{U_{AB}}(\zeta) = \tr_A \left[U_{AB} \left( \frac{\eins}{2} \otimes \zeta \right) U_{AB}^\dagger \right].
\end{equation}
Note that this channel is unital, i.e., $\Gamma_{U_{AB}} (\eins) = \eins$. In the following, we characterize a two-qubit unitary $U_{AB}$ in the context of both preserving and creating entanglement via the moments $\mathcal{P}^{(t)}(\Gamma_{U_{AB}})$ and $\mathcal{C}^{(t)}(U_{AB})$ for $t=2,4$.

For example, in the case with $U_{AB} = \eins_A \otimes \eins_B$ and $U_{AB} = \swap_{AB}$ for $\swap_{AB}$ being the SWAP operator, we notice that $\mathcal{P}^{(t)} (\Gamma_{\eins_A \otimes \eins_B}) = 3$ and $\mathcal{P}^{(t)} (\Gamma_{\swap_{AB}}) = 0$, because $\Gamma_{\eins_A \otimes \eins_B}$ becomes the identity channel and $\Gamma_{\swap_{AB}}$ becomes the completely depolarizing channel. On the other hand, both unitaries are non-entangling and have the same moments $\mathcal{C}^{(t)}(\eins_A \otimes \eins_B) = \mathcal{C}^{(t)}(\swap_{AB}) = 1$. This suggests that $\mathcal{P}^{(t)} (\Gamma_{U_{AB}})$ can distinguish between unitaries that are indistinguishable using $\mathcal{C}^{(t)}(U_{AB})$.

Now, we can present our results as follows:
\begin{observation}   \label{ob:B-C-relation}
    Consider the space spanned by the moments $\mathcal{P}^{(t)}(\Gamma_{U_{AB}})$ and $\mathcal{C}^{(t)}(U_{AB})$ for an arbitrarily two-qubit unitary $U_{AB}$ and the induced single-qubit unital channel $\Gamma_{U_{AB}}$ as in Eq.~(\ref{eq:noisy-unitary}), shown in Fig.~\ref{fig:Pt_Ct}~(a) for $t=2$ and Fig.~\ref{fig:Pt_Ct}~(b) for $t=4$. For any two-qubit unitary $U_{AB}$, if $\Gamma_{U_{AB}}$ is entanglement breaking, then it obeys $F_{\rm EB}^{(t)}[\mathcal{C}^{(t)}({U_{AB}}), \mathcal{P}^{(t)}(\Gamma_{U_{AB}})] \geq 0$ for $t=2,4$, where $F_{\rm EB}^{(t)}$ defines the boundaries of the entanglement-breaking region. A violation of these inequalities, corresponding to a point lying outside the entanglement-breaking regions, identifies a broader class of entanglement-preserving cases, compared to the criteria based solely $\mathcal{P}^{(t)}(\Gamma_{\rm U}) \leq 1$ given in Observation~\ref{ob:EB-channel}.
\end{observation}

\begin{figure}[t]
        \begin{overpic}[width=1.0\linewidth]{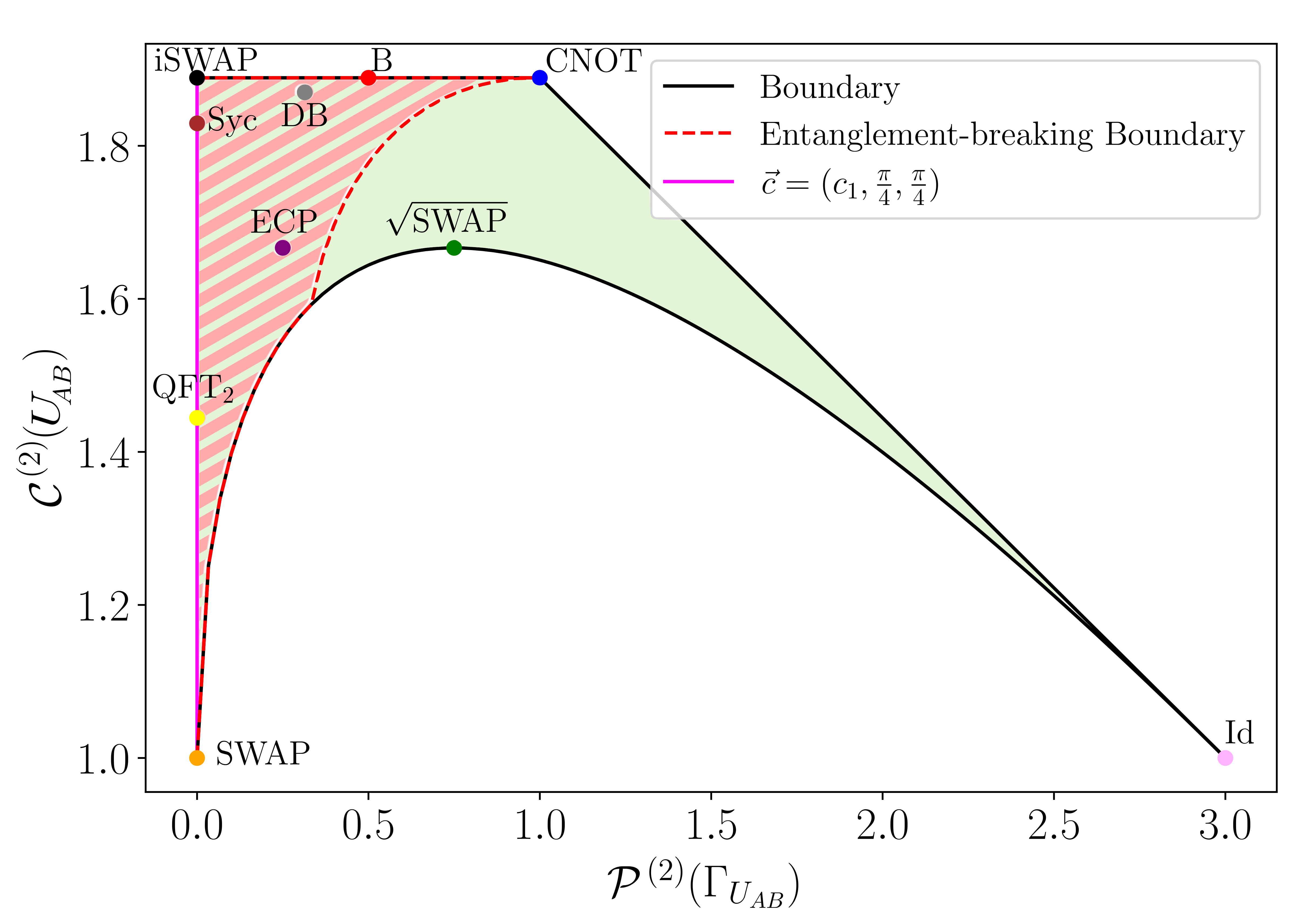}
        \put(2,72){{(a)}} 
        \end{overpic}
    
        \vspace{0.3cm}
    
        \begin{overpic}[width=1.0\linewidth]{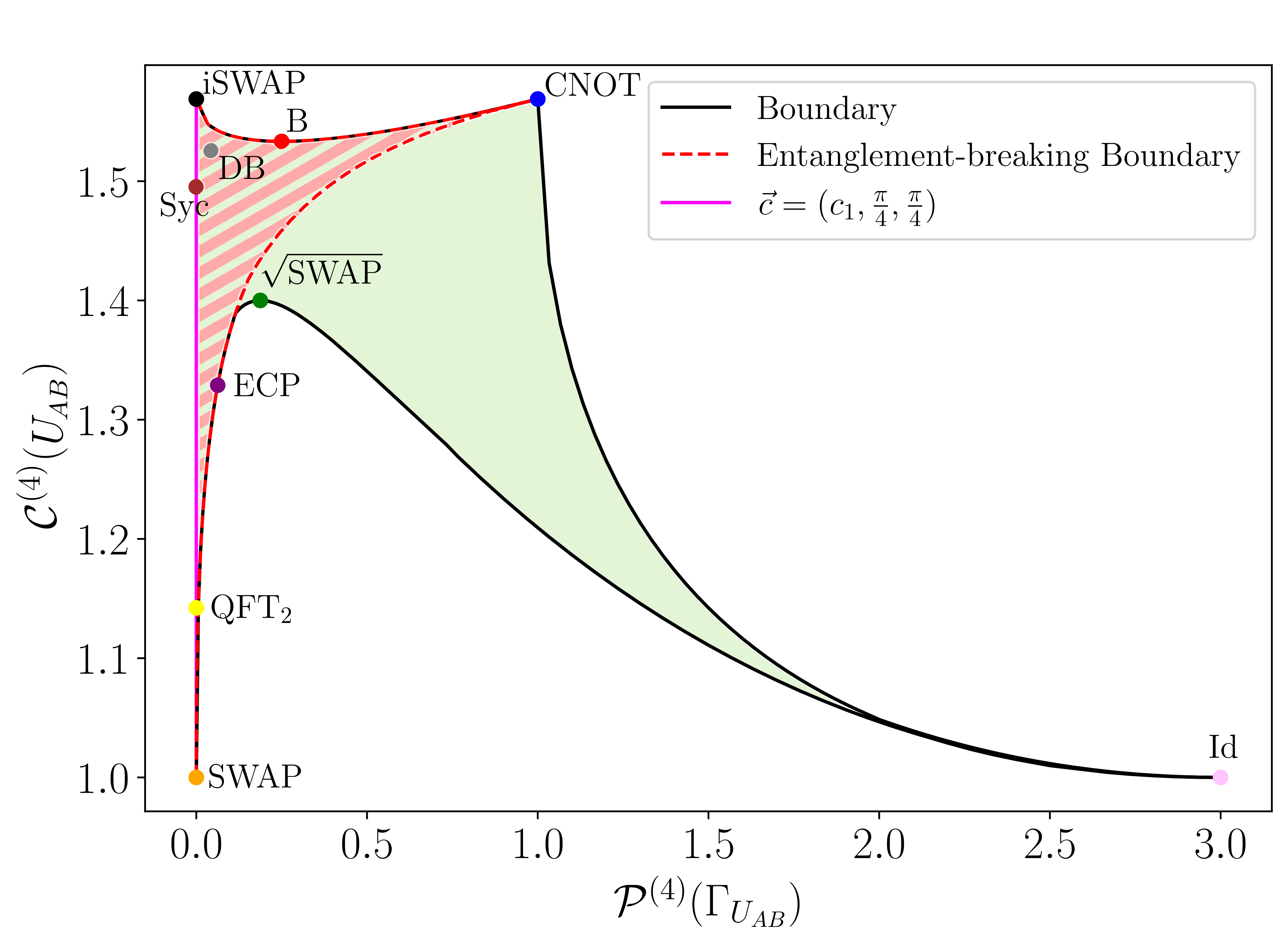}
        \put(2,72){{(b)}}
        \end{overpic}

        \caption{Representation of entanglement-preserving and entanglement-breaking channels in the space spanned by the second moments, $\mathcal{P}^{(2)}(\Gamma_{U_{AB}})$ and $\mathcal{C}^{(2)}({U_{AB}})$ in (a), or the fourth moments, $\mathcal{P}^{(4)}(\Gamma_{U_{AB}})$ and $\mathcal{C}^{(4)}({U_{AB}})$ in (b), for a two-qubit unitary $U_{AB}$ and the induced single-qubit channel $\Gamma_{U_{AB}}$. All entanglement-breaking channels are contained in the hashed red area with $F_{\rm EB}^{(t)}[\mathcal{P}^{(t)}(\Gamma_{U_{AB}}), \mathcal{C}^{(t)}(U_{AB})] \geq 0$, presented in Observation~\ref{ob:B-C-relation}. The green region corresponds to entanglement-preserving cases, which again marks the improvement of the criterion in Observation~\ref{ob:EB-channel}. Several well-known unitaries from Table~\ref{tab:gates} are located.}
        \label{fig:Pt_Ct}
\end{figure}

\begin{proof}
    We begin by recalling that for any two-qubit unitary $U_{AB}$, the relation $\mathcal{C}^{(t)}(U_{AB}) = \mathcal{C}^{(t)}(U_{AB}^{\prime})$ holds, where $U_{AB}^\prime = e^{i \sum_{i=1}^3 c_i \sigma_i \otimes \sigma_i}$ for $c_i \in [0, {\pi}/{4}]$ in Eq.~(\ref{eq:U_{AB}prime}) and $U_{AB}^\prime = V_A \otimes V_B U_{AB} W_A \otimes W_B$ for some unitaries $V_A, V_B, W_A, W_B$. Similarly, we can also show that
    \begin{equation}
        \mathcal{P}^{(t)}(\Gamma_{U_{AB}}) = \mathcal{P}^{(t)}(\Gamma_{U_{AB}^{\prime}})
    \end{equation}
    holds for any unitary $U_{AB}$. This follows from the fact $\Gamma_{U_{AB}^{\prime}}(\rho)= V \Gamma_{U_{AB}} (W \rho W^\dagger) V^\dagger$ for unitaries $V,W$, as stated after Eq.~(\ref{eq:invariance_B}). In fact, using the relation $\tr_A[V_A \otimes V_B X_{AB} V_A^\dagger \otimes V_B^\dagger] = V_B \tr_A[ X_{AB}] V_B^\dagger$ for an operator $X_{AB}$, we have
    \begin{subequations}
    \begin{align}
        \!\!
        \Gamma_{U_{AB}^{\prime}}(\rho)
        &\!=\! \tr_A \! \left[U_{AB}^{\prime}
        \left( \frac{\eins}{2} \otimes \rho \right)
        (U_{AB}^{\prime})^\dagger\right]
        \\
        &\!=\! V_B \tr_A \! \left[U_{AB}
        \left( \frac{\eins}{2} \otimes W_B \rho W_B^\dagger \right)
        U_{AB}^\dagger\right] V_B^\dagger
        \\
        &\!=\! V_B \Gamma_{U_{AB}}(W_B \rho W_B^\dagger) V_B^\dagger.
    \end{align}
    \end{subequations}

    Thus, both $\mathcal{C}^{(t)}(U_{AB})$ and $\mathcal{P}^{(t)}(\Gamma_{U_{AB}})$ can be written in terms of the canonical parameters $\vec{c}=(c_1, c_2, c_3)$. The explicit expressions for $\mathcal{C}^{(t)}(U_{AB})$ are given in Eqs.~(\ref{Eq:EC2},\ref{Eq:EC4}), while one can immediately find that $\mathcal{P}^{(t)}(\Gamma_{U_{AB}}) = \sum_{\alpha=1}^3 D_\alpha^t$ as in Eqs.~(\ref{eq:B2_unital_form},\ref{eq:B4_unital_form}), since $\Gamma_{U_{AB}}$ is unital. Here, $D_\alpha$ are expressed as $D_1 = \cos \left(2 c_1\right) \cos \left(2 c_2\right)$, $D_2 = \cos \left(2 c_1\right) \cos \left(2 c_3\right)$, and $D_3 = \cos \left(2 c_2\right) \cos \left(2 c_3\right)$. Consequently, the whole region spanned by the moments $\mathcal{P}^{(t)}(\Gamma_{U_{AB}})$ and $\mathcal{C}^{(t)}(U_{AB})$ shown in Fig.~\ref{fig:Pt_Ct} can be obtained by maximizing/minimizing $\mathcal{C}^{(t)}(U_{AB})$ over $c_1, c_2, c_3 \in [0, {\pi}/{4}]$ for a fixed value $\mathcal{P}^{(t)}(\Gamma_{U_{AB}})$. The entanglement-breaking boundary is also obtained by imposing the additional constraint $\sum_{\alpha=1}^3 |D_\alpha| \leq 1$ from Eq.~(\ref{eq:EB_constraints}), 
\end{proof}

In Fig.~\ref{fig:Pt_Ct}, the solid black curves represent the whole bounds of the moments of ${U_{AB}}$ and $\Gamma_{U_{AB}}$, while the dashed red curves enclose the entanglement-breaking regions. All channels enclosed by the solid black curves, except local unitaries and the SWAP operator, are entanglement-creating, while the corresponding induced $\Gamma_{U_{AB}}$ are either entanglement-preserving or entanglement-breaking. Several well-known two-qubit unitary gates in Table~\ref{tab:gates} are also shown, positioned according to their canonical parameters $\vec{c} = (c_1, c_2, c_3)$. Among these, only the $\sqrt{\mathrm{SWAP}}$ gate induces an entanglement-preserving channel, while the others induce entanglement-breaking channels.

The pair of moments $\mathcal{P}^{(t)}(\Gamma_{U_{AB}})$ and $\mathcal{C}^{(t)}(U_{AB})$ (for $t=2,4$) together provide a finer characterization of two-qubit unitaries than either moment alone. For example, the $i$SWAP and SWAP gates are distinguishable by $\mathcal{C}^{(t)}(U_{AB})$ but not by $\mathcal{P}^{(t)}(\Gamma_{U_{AB}})$, while the two special perfect entangler, the $i$SWAP and CNOT gates, are distinguishable by $\mathcal{P}^{(t)}(\Gamma_{U_{AB}})$ but not by $\mathcal{C}^{(t)}(U_{AB})$. In fact, the $i$SWAP and SWAP gates are connected by the vertical pink line at $\mathcal{P}^{(t)}(\Gamma_{U_{AB}}) = 0$, corresponding to the family of unitaries with $\vec{c}=(c_1, \pi/4, \pi/4)$ for $c_1 \in [0, \pi/4]$, known as the \emph{parametric SWAP} or $p$SWAP \cite{smith2017}. The induced channel $\Gamma_{p \rm SWAP}$ forms a boundary of the entanglement-breaking region. Also, the horizontal line connecting the $i$SWAP and CNOT gates in Fig.~\ref{fig:Pt_Ct}~(a) is generated by $\vec{c} = (\pi/4,0,c_3)$ for $c_3 \in [0, \pi/4]$. Along this family, $\mathcal{C}^{(2)}(U_{AB})$ is independent of $c_3$ following directly from Eq.~\eqref{Eq:EC2} and thus it is a constant value, while $\mathcal{C}^{(4)}(U_{AB})$ depends on $c_3$, giving rise to the curved trajectory shown in Fig.~\ref{fig:Pt_Ct}~(b).

Finally, we discuss the advantage of using the fourth moments in Fig.~\ref{fig:Pt_Ct}~(b) compared with the second moments in Fig.~\ref{fig:Pt_Ct}~(a). In fact, there are some two-qubit unitaries that can be distinguished by $\mathcal{P}^{(4)}(\Gamma_{U_{AB}})$ but not by both $\mathcal{C}^{(2)}(U_{AB})$ and $\mathcal{P}^{(2)}(\Gamma_{U_{AB}})$. Examples are given by $V_{AB}$ with $\vec{c}_V=(\pi/4, \pi/12,\pi/12)$ and $W_{AB}$ with $\vec{c}_W = (0,\pi/6,\pi/6)$. Although both have the same values
\begin{subequations}
\begin{align}
    \mathcal{C}^{(2)}(V_{AB}) &= \mathcal{C}^{(2)}(W_{AB}) = \frac{11}{6},
    \\
    \mathcal{C}^{(4)}(V_{AB}) &= \mathcal{C}^{(4)}(W_{AB}) = \frac{297}{200},
    \\
    \mathcal{P}^{(2)}(\Gamma_{V_{AB}}) &= \mathcal{P}^{(2)}(\Gamma_{W_{AB}}) = \frac{9}{16},
\end{align}
\end{subequations}
they have different values
\begin{equation}
    \mathcal{P}^{(4)}(\Gamma_{V_{AB}}) = \frac{81}{256},
    \quad
    \mathcal{P}^{(4)}(\Gamma_{W_{AB}}) = \frac{33}{256}.
\end{equation}
That is, they appear as a single point in Fig.~\ref{fig:Pt_Ct}~(a), but are represented as two separate points in Fig.~\ref{fig:Pt_Ct}~(b).

\section{Conclusion}
We have developed methods for characterizing two-qubit channels with respect to their ability to create, preserve, and break entanglement within moment-based frameworks. We first showed that the fourth moment can identify entanglement-creating channels that can be classified as non-entangling based solely on the second moment. Similarly, we found that the fourth moment can detect entanglement-preserving channels that can be recognized as entanglement-breaking by the second moment alone. Finally, we demonstrated that combining two distinct types of moments enables a more refined characterization of two-qubit unitaries than either type can provide on its own.

Several directions for further research remain. First, it would be interesting to extend our approach to multipartite particles and higher-dimensional systems by employing the sector lengths \cite{Wyderka_2020,Eltschka2020} or the trace polynomials of the generalized Bloch correlation matrix \cite{Imai_2021, Wyderka_2023}. Second, our results motivate further investigation into other resource-creating and resource-breaking properties beyond entanglement \cite{Chitambar_2019}. Finally, to enhance the practical applicability of our framework in experimental settings, it would be worthwhile to examine the statistical treatment of finite data.

\textit{Note added:} While finishing this manuscript, we became aware of a related work \cite{rudzinski2026}.

\section{Acknowledgments}
We thank
Daniel Heineken,
Martin Kliesch,
Nikolai Miklin, and
Stefan Nimmrichter
for discussions.
S.S. is funded by Fujitsu Germany GmbH and Dataport and the Hamburg Quantum Computing project, which is co-financed by the ERDF of the European Union and the Fonds of the Hamburg Ministry of Science, Research, Equalities and Districts (BWFGB).
S.I. acknowledges support from JST ASPIRE (JPMJAP2339).

\vspace{0.25cm}
{\footnotesize
\hypertarget{email1}{}\noindent\textsuperscript{*} \href{mailto:salwa.shaglel@tuhh.de}{salwa.shaglel@tuhh.de} \\
\hypertarget{email2}{}\textsuperscript{\textdagger} \href{mailto:satoyaimai@yahoo.co.jp}{satoyaimai@yahoo.co.jp}
}

\clearpage
\newpage
\appendix
\onecolumngrid

\section{Derivation of Eqs.~(\ref{Eq:C^(2)}, \ref{Eq:C^(4)})}\label{ap:derivation_details_ENT_CRT}
Here we derive Eqs.~(\ref{Eq:C^(2)}, \ref{Eq:C^(4)}) presented in the main text by showing that, for any two-qubit channel $\Lambda(\rho_{AB})=\sum_{\alpha} K_\alpha \rho_{AB} K^\dagger_\alpha$, the moments $\mathcal{C}^{(2)} (\Lambda)$ and $\mathcal{C}^{(4)} (\Lambda)$ are given by
\begin{subequations}
    \begin{align} \label{eq:ap:C2Lambda}
    \mathcal{C}^{(2)} (\Lambda)
    &= \sum_{i,j} \sum_{\alpha, \beta}
    \tr \! \left[
    K_{\alpha \beta}
    P_A^{(2)} \otimes P_B^{(2)}
    K_{\alpha \beta}^\dagger
    (\sigma_i^A \! \otimes \! \sigma_j^B \! \otimes \! 
    \sigma_i^A \! \otimes \! \sigma_j^B)
    \right],    
    \\
    \label{eq:ap:C4Lambda}
    \mathcal{C}^{(4)} (\Lambda)
    &= \sum_{\substack{i,j \\ k,l}}
    \sum_{\substack{\alpha, \beta \\ \gamma, \delta}}
    \tr \! \left[
    K_{\alpha \beta \gamma \delta}
    P_A^{(4)} \otimes P_B^{(4)}
    K_{\alpha \beta \gamma \delta}^\dagger
    (\sigma_i^A \! \otimes \! \sigma_j^B  \! \otimes \!
    \sigma_k^A \! \otimes \! \sigma_j^B \! \otimes \!
    \sigma_k^A \! \otimes \!  \sigma_l^B \! \otimes \!
    \sigma_i^A \! \otimes \! \sigma_l^B )
    \right],
    \end{align}
\end{subequations}
where $K_{\alpha \beta} \vcentcolon = K_\alpha \otimes K_\beta$, $K_{\alpha \beta \gamma \delta} \vcentcolon = K_\alpha \otimes K_\beta \otimes K_\gamma \otimes K_\delta$, and $P_X^{(t)} = [1/(t+1)] \Tilde{P}_X^{(t)}$ for $X=A,B$ and $\Tilde{P}_X^{(t)}$ representing the projector onto the symmetric subspace of $t$ qubits.

\begin{proof}
We begin by recalling that $\mathcal{C}^{(t)} (\Lambda)$ is defined as
\begin{equation}
    \mathcal{C}^{(t)} (\Lambda)
    = \int d\psi_A \int d\psi_B \,
    I_t[\Lambda (\ket{\psi_A} \otimes \ket{\psi_B})], \label{eq:C^(t)(Lambda)}
\end{equation}
where $I_2 (\rho_{AB}) = \sum_{i,j=1,2,3} T_{i,j}^2$ and $I_4 (\rho_{AB}) = \sum_{i,j,k,l=1,2,3} T_{i,j} T_{k,j} T_{k,l} T_{i,l}$ with $T_{i,j} = \tr(\rho_{AB} \sigma_i^A \otimes \sigma_j^B)$. Now we consider the case in Eq.~(\ref{eq:setting_rho_AB}) where $\rho_{AB} = \Lambda (\ket{\psi_{AB}})$ and $\ket{\psi_{AB}} = \ket{\psi_A} \otimes \ket{\psi_B}$. To proceed, let us rewrite the form of $I_t[\Lambda (\ket{\psi_A} \otimes \ket{\psi_B})]$. Using the fact that $\tr(X) \tr(Y) = \tr(X \otimes Y)$ for operators $X,Y$, we find
\begin{subequations}
    \begin{align}
    \! \! \!
    I_2[\Lambda (\ket{\psi_A} \otimes \ket{\psi_B})]
    &\!=\! \sum_{i,j} \sum_{\alpha, \beta}
    \tr \! \left[
    K_{\alpha \beta}
    (\ket{\psi_A}\! \bra{\psi_A}
    \! \otimes \!
    \ket{\psi_B}\! \bra{\psi_B})^{\otimes 2}
    K_{\alpha \beta}^\dagger
    (\sigma_i^A \! \otimes \! \sigma_j^B \! \otimes \! 
    \sigma_i^A \! \otimes \! \sigma_j^B)
    \right],
    \\
    \! \! \!
    I_4[\Lambda (\ket{\psi_A} \otimes \ket{\psi_B})]
    &\!=\! \! \sum_{\substack{i,j \\ k,l}}
    \sum_{\substack{\alpha, \beta \\ \gamma, \delta}}
    \tr \! \left[
    K_{\alpha \beta \gamma \delta}
    (\ket{\psi_A}\! \bra{\psi_A}
    \! \otimes \!
    \ket{\psi_B}\! \bra{\psi_B})^{\otimes 4}
    K_{\alpha \beta \gamma \delta}^\dagger
    (\sigma_i^A \! \otimes \! \sigma_j^B  \! \otimes \!
    \sigma_k^A \! \otimes \! \sigma_j^B \! \otimes \!
    \sigma_k^A \! \otimes \!  \sigma_l^B \! \otimes \!
    \sigma_i^A \! \otimes \! \sigma_l^B )
    \right].
    \end{align}
\end{subequations}

Next, inserting the above expressions into $\mathcal{C}^{(t)} (\Lambda)$ and exchanging the integrals and trace, we obtain
\begin{subequations}
    \begin{align}
    \mathcal{C}^{(2)} (\Lambda)
    &= \sum_{i,j} \sum_{\alpha, \beta}
    \tr \! \left[
    K_{\alpha \beta}
    \mathcal{I}_{2}^{AB}
    K_{\alpha \beta}^\dagger
    (\sigma_i^A \! \otimes \! \sigma_j^B \! \otimes \! 
    \sigma_i^A \! \otimes \! \sigma_j^B)
    \right],
    \\
    \mathcal{C}^{(4)} (\Lambda)
    &= \sum_{\substack{i,j \\ k,l}}
    \sum_{\substack{\alpha, \beta \\ \gamma, \delta}}
    \tr \! \left[
    K_{\alpha \beta \gamma \delta}
    \mathcal{I}_{4}^{AB}
    K_{\alpha \beta \gamma \delta}^\dagger
    (\sigma_i^A \! \otimes \! \sigma_j^B  \! \otimes \!
    \sigma_k^A \! \otimes \! \sigma_j^B \! \otimes \!
    \sigma_k^A \! \otimes \!  \sigma_l^B \! \otimes \!
    \sigma_i^A \! \otimes \! \sigma_l^B )
    \right],
    \end{align}
\end{subequations}
where $\mathcal{I}_{t}^{AB} \vcentcolon = \int d\psi_A \int d\psi_B \, (\ket{\psi_A}\! \bra{\psi_A} \! \otimes \! \ket{\psi_B}\! \bra{\psi_B})^{\otimes t}$. Letting $P_X^{(t)} = \int d\psi_X \, (\ket{\psi_X}\! \bra{\psi_X})^{\otimes t}$ for $X=A,B$ acting on a $t$-qubit system, we have that $\mathcal{I}_{t}^{AB} = P_A^{(t)} \otimes P_B^{(t)}$. Hence we can arrive at Eqs.~(\ref{eq:ap:C2Lambda}, \ref{eq:ap:C4Lambda}).

Finally, we notice that $P_X^{(t)} = [1/(t+1)] \Tilde{P}_X^{(t)}$ for $X=A,B$, and $\Tilde{P}_X^{(t)}$ is the projector onto the \textit{symmetric subspace} of $t$ qubits \cite{Roberts_2017}. Specifically, in this case, the symmetric subspace has dimension $t+1$ and is spanned by the $t$-qubit Dicke states $\{\ket{D^{k}_{t}}\}_{k=0}^{t}$ defined as
\begin{equation}
    \ket{D^{k}_{t}} = \binom{t}{k}^{-1/2} \sum_{\pi_{t}} \pi_{t} (\ket{0}^{\otimes(t-k)} \ket{1}^{\otimes k}),
\end{equation}
where $k$ denotes the number of excitations and the sum runs over all possible permutations $\pi_{t}$ acting on $t$ qubits. Thus, $P_X^{(t)}$ is given by $P_X^{(t)} = [1/(t+1)]  \sum_{k=0}^t \ket{D^{k}_{t}} \! \bra{D^{k}_{t}}$.
\end{proof}

\section{Derivation of Eqs.~(\ref{Eq:EC2}, \ref{Eq:EC4})}\label{ap:derivation_details_ENT_CRT_noisy_case}
Here we derive Eqs.~(\ref{Eq:EC2}, \ref{Eq:EC4}) presented in the main text by showing that, for the two-qubit unitary $U_{AB} = e^{i \sum_{i=1}^3 c_i \sigma_i \otimes \sigma_i}$ with $c_i \in [0, \pi/4]$, the moments $\mathcal{C}^{(2)} (U_{AB})$ and $\mathcal{C}^{(4)} (U_{AB})$ are given by
\begin{subequations}
    \begin{align}\label{eq:ap:C2Lambdap}
    \mathcal{C}^{(2)} (U_{AB})
    &= \frac{1}{9} \Bigl\{ 15 - \cos(4c_{12}^-) - \cos(4c_{12}^+)
    - 2 \bigl[\cos(4c_1) + \cos(4c_2)\bigr] \cos(4c_3)\Bigl\},
    \\
    \label{eq:ap:C4Lambdap}
    \mathcal{C}^{(4)} (U_{AB})
    &= \frac{1}{900} \Bigl\{ 1218+4 \cos(8 c_1)-76 \cos(4c_{12}^-) +9 \cos(8c_{12}^-)+4 \cos(8 c_2) -76 \cos(4c_{12}^+) 
    \nonumber \\
    &+9 \cos(8c_{12}^+) + 8 \bigl[\cos(4 c_1)+\cos(4 c_2)\bigl] \bigl[-22 +6 \cos(4 c_1) \cos(4 c_2) \bigl] \cos(4 c_3)
    \nonumber\\
    &+ 4 \bigl[1+3 \cos(4c_{12}^-) \bigl] \bigl[1 +3 \cos(4c_{12}^+) \bigl] \cos(8 c_3)\Bigl\}.
    \end{align}
\end{subequations}
where $c_{12}^+ = c_1+c_2$ and $c_{12}^- = c_1-c_2$.

\begin{proof}
We begin by writing the following relation:
\begin{equation}
    U_{AB}
    =\exp \left(i \sum_{i=1}^3 c_i \sigma_i^A \otimes \sigma_i^B \right)
    = \sum_{a=0}^3 m_a \sigma_a^A \otimes \sigma_a^B,
    \label{eq:canonical_form_unitary}
\end{equation}
where we explicitly denoted the superscripts $A$ and $B$ in $\sigma_i^A \otimes \sigma_i^B$ and 
\begin{subequations}
    \begin{align}
    m_0 &= \cos(c_1)\cos(c_2)\cos(c_3) + i \sin(c_1)\sin(c_2)\sin(c_3), \\
    m_1 &= \cos(c_1)\sin(c_2)\sin(c_3) + i \sin(c_1) \cos(c_2) \cos(c_3), \\
    m_2 &= \sin(c_1)\cos(c_2)\sin(c_3) + i \cos(c_1)\sin(c_2)\cos(c_3),\\
    m_3 &= \sin(c_1)\sin(c_2)\cos(c_3) + i \cos(c_1) \cos(c_2) \sin(c_3). 
    \end{align}
    \label{Eq:m_to_c}
\end{subequations}
Inserting the above expressions into $I_t[U_{AB} \ket{\psi_A} \otimes \ket{\psi_B}]$, a direct calculation leads to
\begin{subequations}
    \begin{align} 
    \mathcal{C}^{(2)} (U_{AB})
    &\!=\! \sum_{i,j} \!
     \!
    \sum_{\substack{a_1,a_2, \\ b_1,b_2}} \!
    M_{a_1 a_2 b_1 b_2}
    \tr \! \left[
    (\sigma_{a_1 a_1 a_2 a_2})
    P_A^{(2)} \otimes P_B^{(2)}
    (\sigma_{b_1 b_1 b_2 b_2})
    (\sigma_{ijij})
    \right],    
    \\
    \mathcal{C}^{(4)} (U_{AB})
    &\!=\! \sum_{\substack{i,j \\ k,l}} \!
    \sum_{\substack{a_1,a_2, \\ b_1,b_2}} \!
    \sum_{\substack{a_3,a_4, \\ b_3,b_4}} \!
    M_{a_1 a_2 a_3 a_4 b_1 b_2 b_3 b_4}
    \tr \! \left[
    (\sigma_{a_1 a_1 a_2 a_2 a_3 a_3 a_4 a_4})
    P_A^{(4)} \otimes P_B^{(4)}
    (\sigma_{b_1 b_1 b_2 b_2 b_3 b_3 b_4 b_4})
    (\sigma_{ijkjklil})
    \right],
    \end{align}
\end{subequations}
where, for the sake of simplicity, we denoted that
\begin{subequations}
    \begin{align}
    M_{a_1 a_2 b_1 b_2}
    &\vcentcolon = m_{a_1} m_{a_2} m_{b_1}^* m_{b_2}^*,
    \\
    M_{a_1 a_2 a_3 a_4 b_1 b_2 b_3 b_4}
    &\vcentcolon = m_{a_1} m_{a_2} m_{a_3} m_{a_4} m_{b_1}^* m_{b_2}^* m_{b_3}^* m_{b_4}^*,
    \\
    \sigma_{abcd}
    &\vcentcolon =
    (\sigma_a^A \otimes \sigma_b^B) \otimes (\sigma_c^A \otimes \sigma_d^B),
    \\
    \sigma_{abcdefgh}
    &\vcentcolon =
    (\sigma_a^A \otimes \sigma_b^B) \otimes (\sigma_c^A \otimes \sigma_d^B)
    \otimes
    (\sigma_e^A \otimes \sigma_f^B) \otimes (\sigma_g^A \otimes \sigma_h^B).
    \end{align}
\end{subequations}

Moreover, we can further simplify the form of $\mathcal{C}^{(t)} (U_{AB})$ by splitting the trace into each subsystem $A$ and $B$. Setting
\begin{subequations}
    \begin{align}
    \label{eq:Paabb2}
    \Xi_{a_1 a_2 b_1 b_2, i i}^X
    &\vcentcolon = \tr \! \left[
    (\sigma_{a_1}^X \otimes \sigma_{a_2}^X)
    P_X^{(2)}
    (\sigma_{b_1}^X \otimes \sigma_{b_2}^X)
    (\sigma_{i}^X \otimes \sigma_{i}^X)
    \right],
    \\
    \Xi_{a_1 a_2 a_3 a_4 b_1  b_2 b_3 b_4, i j k l}^X
    &\vcentcolon = \tr \! \left[
    (\sigma_{a_1}^X \otimes \sigma_{a_2}^X \otimes \sigma_{a_3}^X \otimes \sigma_{a_4}^X)
    P_X^{(4)}
    (\sigma_{b_1}^X \otimes \sigma_{b_2}^X \otimes \sigma_{b_3}^X \otimes \sigma_{b_4}^X)
    (\sigma_{i}^X \otimes \sigma_{j}^X \otimes \sigma_{k}^X \otimes \sigma_{l}^X)
    \right],
    \label{eq:Paaaabbbb4}
    \end{align}
\end{subequations}
We can obtain
\begin{subequations}
    \begin{align} 
    \mathcal{C}^{(2)} (U_{AB})
    &\!=\! \sum_{i,j} \!
    \sum_{\substack{a_1,a_2, \\ b_1,b_2}} \!
    M_{a_1 a_2 b_1 b_2}
    \Xi_{a_1 a_2 b_1 b_2, i i}^A
    \Xi_{a_1 a_2 b_1 b_2, j j}^B,
    \\
    \mathcal{C}^{(4)} (U_{AB})
    &\!=\! \sum_{\substack{i,j \\ k,l}} 
    \sum_{\substack{a_1,a_2, \\ b_1,b_2}} \!
    \sum_{\substack{a_3,a_4, \\ b_3,b_4}} \!
    M_{a_1 a_2 a_3 a_4 b_1 b_2 b_3 b_4}
    \Xi_{a_1 a_2 a_3 a_4 b_1 b_2 b_3 b_4, i k k i}^A
    \Xi_{a_1 a_2 a_3 a_4 b_1 b_2 b_3 b_4, j j l l}^B.
    \end{align}
\end{subequations}
Finally, inserting the form of $P_X^{(t)} = [1/(t+1)] \sum_{k=0}^t \ket{D^{k}_{t}} \! \bra{D^{k}_{t}}$ into Eqs.~(\ref{eq:Paabb2}, \ref{eq:Paaaabbbb4}), one can straightforwardly evaluate the moments. After summarizing terms, we can arrive at Eqs.~(\ref{eq:ap:C2Lambdap}, \ref{eq:ap:C4Lambdap}).
\end{proof}

\section{Derivation of Eqs.~(\ref{Eq:E^(2)}, \ref{Eq:E^(4)})}\label{ap:derivation_E^(2)andE^(4)}
Here we derive Eqs.~(\ref{Eq:E^(2)}, \ref{Eq:E^(4)}) presented in the main text by showing that, for the single-qubit channel $\Gamma(\rho) = \sum_\alpha E_\alpha (\rho) E_\alpha^\dagger$, the moments $\mathcal{P}^{(2)} (\Gamma)$ and $\mathcal{P}^{(4)} (\Gamma)$ are given by
\begin{subequations}
    \begin{align}
        \mathcal{P}^{(2)} (\Gamma)
        &= \frac{1}{4} \sum_{i,j} \sum_{\alpha , \beta} \tr \left[  E_{\alpha \beta} \Phi_2(\sigma_i \otimes \sigma_i) E^\dagger_{\alpha \beta} (\sigma_j \otimes \sigma_j) \right], 
        \label{eq:app_E^(2)}
        \\
        \mathcal{P}^{(4)} (\Gamma)
        &= \frac{1}{16} \sum_{\substack{i,j \\ k,l}} \sum_{\substack{\alpha, \beta \\ \gamma, \delta}} \tr \left[  E_{\alpha \beta \gamma \delta} \Phi_4(\sigma_i \otimes \sigma_k \otimes \sigma_k \otimes \sigma_i) E^\dagger_{\alpha \beta \gamma \delta} (\sigma_j \otimes \sigma_j \otimes \sigma_l \otimes \sigma_l) \right],  
        \label{eq:app_E^(4)}
    \end{align}
\end{subequations}
where $E_{\alpha \beta} \vcentcolon = E_\alpha \otimes E_\beta$, $E_{\alpha \beta \gamma \delta} \vcentcolon = E_\alpha \otimes E_\beta \otimes E_\gamma \otimes E_\delta$, and $\Phi_t (\mathcal{O}) \vcentcolon = \int dU \,    U^{\otimes t} \mathcal{O} (U^\dagger)^{\otimes t}$ for an operator $\mathcal{O}$.

\begin{proof}
    We begin by recalling that the moments $\mathcal{P}^{(2)}(\Gamma)$ and $\mathcal{P}^{(4)}(\Gamma)$ are rewritten as
    \begin{subequations}
        \begin{align}
        \mathcal{P}^{(2)} (\Gamma)
        &= \int dU \,
        I_2[({\rm id} \otimes \Gamma) (\ket{\Psi^-_{U}}\! \bra{\Psi^-_{U}})],
        \\
        \mathcal{P}^{(4)} (\Gamma)
        &= \int dU \,
        I_4[({\rm id} \otimes \Gamma) (\ket{\Psi^-_{U}}\! \bra{\Psi^-_{U}})],
        \end{align}
    \end{subequations}
    where $I_t$ with $t=2,4$ are defined after Eq.~\eqref{eq:C^(t)(Lambda)} and $\ket{\Psi_U} = (\eins \otimes U) \ket{\Psi^-}$ with $\ket{\Psi^-} = (1/\sqrt{2})(\ket{0}\otimes\ket{1} - \ket{1}\otimes\ket{0})$. To proceed, let us express the state $\ket{\Psi^-_U}$ in the Pauli basis as
    \begin{equation}
        \ket{\Psi^-_U}\bra{\Psi^-_U}
        = \frac{1}{4}  \left(\eins \otimes \eins - \sum_{a=1,2,3} \sigma_a \, \otimes \, U\sigma_a U^\dagger \right).
    \end{equation}    
    This directly leads to
    \begin{equation}
        I_t \left[ ({\rm id} \otimes \Gamma) \bigl( \ket{\Psi^-_U}\bra{\Psi^-_U} \bigr) \right]
        = \frac{1}{4^t} I_t \left[({\rm id} \otimes \Gamma)
        \biggl( \sum_{a} \sigma_a \otimes U \sigma_a U^\dagger
        \biggr) \right].        
    \end{equation}
    Using the fact that $\tr(X) \tr(Y) = \tr(X \otimes Y)$ for operators $X,Y$, we find
    \begin{subequations}
        \begin{align}
            I_2 \left[ ({\rm id} \otimes \Gamma) \bigl( \ket{\Psi^-_U}\bra{\Psi^-_U} \bigr) \right]
            &= \frac{1}{4^2}\sum_{i,j} \sum_{a,b} \sum_{\alpha , \beta}
            \tr \left(  \sigma_{ab} \sigma_{ii} \right)
            \tr \left[  E_{\alpha \beta} U^{\otimes 2} \sigma_{ab} (U^\dagger)^{\otimes 2} E^\dagger_{\alpha \beta} \sigma_{jj} \right],
            \\
            I_4 \left[ ({\rm id} \otimes \Gamma) \bigl( \ket{\Psi^-_U}\bra{\Psi^-_U} \bigr) \right]
            &= \frac{1}{16^2}\sum_{\substack{i,j \\ k,l}}
            \sum_{\substack{a,b \\ c,d}}
            \sum_{\substack{\alpha, \beta \\ \gamma, \delta}}
            \tr \left(  \sigma_{abcd} \sigma_{ikki} \right)
            \tr \left[  E_{\alpha \beta \gamma \delta} U^{\otimes 4} \sigma_{abcd} (U^\dagger)^{\otimes 4} E^\dagger_{\alpha \beta \gamma \delta} \sigma_{jjll} \right],
        \end{align}
    \end{subequations}
    where we denoted that $\sigma_{ij} = \sigma_i \otimes \sigma_j$ and $\sigma_{ijkl} = \sigma_i \otimes \sigma_j \otimes \sigma_k \otimes \sigma_l$. By employing the formulas $\tr(\sigma_{ab} \sigma_{ii}) = 4 \delta_{ai} \delta_{bi}$ and $\tr (\sigma_{abcd} \sigma_{ikki}) = 16 \delta_{ai} \delta_{bk} \delta_{ck} \delta_{di}$ for $a,b,c,d,i,k =1,2,3$, and by denoting $\Phi_2(\sigma_{ii}) = \int \, dU \, U^{\otimes 2} \sigma_{ii} (U^\dagger)^{\otimes 2}$ and $\Phi_4(\sigma_{ikki}) = \int \, dU \, U^{\otimes 4} \sigma_{ikki} (U^\dagger)^{\otimes 4}$, we can arrive at Eqs.~(\ref{eq:app_E^(2)}, \ref{eq:app_E^(4)}).
\end{proof}

\section{Derivation of Eqs.~(\ref{eq:B2_unital_form}, \ref{eq:B4_unital_form})}\label{ap:derivation_B2andB4_unital_form}
Here we derive Eqs.~(\ref{eq:B2_unital_form}, \ref{eq:B4_unital_form}) presented in the main text by showing that, for the unital single-qubit channel $\Gamma_{\rm U}(\rho) = \sum_{\alpha=0}^3 d_\alpha \sigma_\alpha \rho \sigma_\alpha$, the moments $\mathcal{P}^{(2)} (\Gamma_{\rm U})$ and $\mathcal{P}^{(4)} (\Gamma_{\rm U})$ are given by
\begin{subequations}
    \begin{align}
    \label{eq:B2form_ap}
    \mathcal{P}^{(2)} (\Gamma_{\rm U})
    = \sum_{\alpha=1}^3  D_\alpha^2, 
    \\
    \label{eq:B4form_ap}
     \mathcal{P}^{(4)} (\Gamma_{\rm U})
    =  \sum_{\alpha=1}^3 D_\alpha^4,
    \end{align}
\end{subequations}
where $D_\alpha = 2( d_\alpha +  d_0) - 1$ and $d_\alpha \in [0,1]$.

\begin{proof}
    We begin by inserting $E_\alpha = \sqrt{d_\alpha} \sigma_\alpha$ into the expressions in Eqs.~(\ref{eq:app_E^(2)}, \ref{eq:app_E^(4)}). By rearranging terms inside the trace with the cyclic property, we obtain
    \begin{subequations}
        \begin{align}
            \mathcal{P}^{(2)} (\Gamma_{\rm U})
            &= \frac{1}{4} \sum_{i,j}
            \tr \left[ \Phi_2(\sigma_i \otimes \sigma_i) \, \chi_2(\sigma_j \otimes \sigma_j) \right], 
            \\
            \mathcal{P}^{(4)} (\Gamma_{\rm U})
            &= \frac{1}{16} \sum_{\substack{i,j \\ k,l}}
            \tr \left[  \Phi_4(\sigma_i \otimes \sigma_k \otimes \sigma_k \otimes \sigma_i) \, \chi_4(\sigma_j \otimes \sigma_j \otimes \sigma_l \otimes \sigma_l) \right].  
        \end{align}
    \end{subequations}    
    where we denoted
    \begin{subequations}
        \begin{align}
            \chi_2(\sigma_j \otimes \sigma_j)
            &\vcentcolon = \sum_{\alpha , \beta}
            d_\alpha d_\beta \,
            (\sigma_\alpha \otimes \sigma_\beta)
            (\sigma_j \otimes \sigma_j)
            (\sigma_\alpha \otimes \sigma_\beta)
            \\
            \chi_4 (\sigma_j \otimes \sigma_j \otimes \sigma_l \otimes \sigma_l)
            &\vcentcolon = \sum_{\substack{\alpha, \beta \\ \gamma, \delta}}
            d_\alpha d_\beta d_\gamma d_\delta \,
            (\sigma_\alpha \otimes \sigma_\beta \otimes \sigma_\gamma \otimes \sigma_\delta)
            (\sigma_j \otimes \sigma_j \otimes \sigma_l \otimes \sigma_l)
            (\sigma_\alpha \otimes \sigma_\beta \otimes \sigma_\gamma \otimes \sigma_\delta).
        \end{align}
    \end{subequations}
    Using the formula $(X_1 \otimes Y_1) (X_2 \otimes Y_2) = X_1 X_2 \otimes Y_1 Y_2$ for operators $X_1, X_2, Y_1, Y_2$ and the fact that $\sigma_m \sigma_n \sigma_m = (1-\delta_{m0})(2\delta_{mn} \sigma_m - \sigma_n) + \delta_{m0} \sigma_n$ for $m=0,1,2,3$ and $n=1,2,3$, we obtain
    \begin{subequations}
        \begin{align}
            \chi_2(\sigma_j \otimes \sigma_j)
            &=  D_j^2 \sigma_j \otimes \sigma_j,
            \\
            \chi_4 (\sigma_j \otimes \sigma_j \otimes \sigma_l \otimes \sigma_l)
            &= D_j^2 D_l^2 \sigma_j \otimes \sigma_j \otimes \sigma_l \otimes \sigma_l,
        \end{align}
    \end{subequations}
    where we denoted $D_i \vcentcolon = 2( d_i +  d_0) - 1$. Substituting these results to $\mathcal{P}^{(2)} (\Gamma_{\rm U})$ and $ \mathcal{P}^{(4)} (\Gamma_{\rm U})$ yields
    \begin{subequations}
        \begin{align}
            \mathcal{P}^{(2)} (\Gamma_{\rm U})
            &= \frac{1}{4} \sum_{i,j}  D_j^2
            \tr \left[ \Phi_2(\sigma_i \otimes \sigma_i) (\sigma_j \otimes \sigma_j) \right], \label{eq:app_P^(2)}
            \\
            \mathcal{P}^{(4)} (\Gamma_{\rm U})
            &= \frac{1}{16} \sum_{\substack{i,j \\ k,l}} D_j^2 D_l^2
            \tr \left[  \Phi_4(\sigma_i \otimes \sigma_k \otimes \sigma_k \otimes \sigma_i) (\sigma_j \otimes \sigma_j \otimes \sigma_l \otimes \sigma_l) \right]. \label{eq:app_P^(4)}
        \end{align}
    \end{subequations}
    
    Finally, we note the following identities provided by the supplemental material of Ref.~\cite{Wyderka2023}:
    \begin{subequations}
        \begin{align}
            \tr \left[ \Phi_2(\sigma_a \otimes \sigma_b) (\sigma_j \otimes \sigma_j) \right]
            &= \frac{4}{3} \delta_{ab}, 
            \\
            \tr \left[  \Phi_4(\sigma_a \otimes \sigma_b \otimes \sigma_c \otimes \sigma_d) (\sigma_j \otimes \sigma_j \otimes \sigma_j \otimes \sigma_j) \right] &= \frac{16}{15} \left( \delta_{ab} \delta_{cd} + \delta_{ac} \delta_{bd} + \delta_{ad} \delta_{cb} \right).
        \end{align}
    \end{subequations}
    Then we have that $\tr \left[ \Phi_2(\sigma_i \otimes \sigma_i) (\sigma_j \otimes \sigma_j) \right] = 4/3$ and 
    $\tr \left[  \Phi_4(\sigma_i \otimes \sigma_k \otimes \sigma_k \otimes \sigma_i) (\sigma_j \otimes \sigma_j \otimes \sigma_l \otimes \sigma_l) \right] = ({16}/{15}) (2\delta_{ik} + 1)$, where the trace vanishes for $l \neq j$. Inserting this into Eqs.~(\ref{eq:app_P^(2)}, \ref{eq:app_P^(4)}), we can arrive at Eqs.~(\ref{eq:B2form_ap}, \ref{eq:B4form_ap}).
\end{proof}

\section{Optimization details of Figs.~\ref{fig:EC_plot}, \ref{fig:EB_plot}, and \ref{fig:Pt_Ct}}\label{ap:optimization_details}
The optimization procedures implemented in Figs.~\ref{fig:EC_plot}, \ref{fig:EB_plot}, and \ref{fig:Pt_Ct} were performed using the \texttt{trust-constr} algorithm implemented in the \texttt{SciPy} \texttt{minimize} function \cite{virtanen2020scipy, scipy_minimize_docs}. The solver was executed with numerical tolerances: the step-size tolerance, gradient-norm tolerance, and barrier tolerance, all set to $10^{-12}$, while the maximum number of iterations is $10^{4}$. These parameters, respectively, control the magnitude of the final optimization step, the accuracy of the first-order optimality conditions, and the precision with which inequality constraints are satisfied. The reported results correspond mostly to the cases in which the convergence criteria were met before reaching the maximum iteration limit; otherwise, the precision achieved may be lower than the chosen tolerance.

\twocolumngrid
\bibliography{ref.bib}
\end{document}